\documentclass[lettersize,journal]{IEEEtran}

\usepackage{amsmath,amsfonts}
\usepackage{algorithmic}
\usepackage{algorithm}
\usepackage{array}
\usepackage[caption=false,font=normalsize,labelfont=sf,textfont=sf]{subfig}
\usepackage{textcomp}
\usepackage{stfloats}
\usepackage{url}
\usepackage{verbatim}
\usepackage{graphicx}
\usepackage{cite}
\usepackage{threeparttable}
\usepackage{booktabs}
\usepackage{nicematrix}

\usepackage{threeparttable}
\usepackage{color}
\usepackage{amssymb}
\usepackage{mathrsfs}
\newtheorem{prop}{Proposition}%[section]
\newcommand{\dtau}{\Delta\tau}
\newcommand{\dep}{\Delta\varepsilon}
\newcommand{\ve}{\varepsilon}

\newcommand{\vp}{\varphi}
\newcommand{\ii}{\mathrm{i}}

\begin{document}

\title{ST-DDA: Dynamic Channel Estimation in the Doppler-Delay-Angle Domain via Sparse Subspace Tracking for TDD Systems}

\author{Xu Zhu and Tiejun Li
        % <-this % stops a space
% <-this % stops a space
\thanks{X. Zhu and T. Li are with Laboratory of Mathematics and Applied Mathematics, School of Mathematical Sciences, and Center for Machine Learning Research, Peking University, Beijing 100871, P.R. China}
\thanks{Email: xuzhu@pku.edu.cn (X. Zhu), tieli@pku.edu.cn (T. Li)}
\thanks{Corresponding author: Tiejun Li}
}


% \IEEEpubid{%0000--0000/00\$00.00~\copyright~2021 IEEE
% pubid\ \ \ \ \ \ \ }

\maketitle

\begin{abstract}

Accurate full-band channel acquisition in frequency-hopping time-division duplex (TDD) systems is challenging because each sounding slot observes only a limited frequency subband, while conventional single-slot recovery cannot fully exploit historical observations. We propose ST-DDA, an online sparse-subspace tracking framework for latest-slot reconstruction in the Doppler--delay--angle (DDA) domain. We first show that the Doppler-domain representation remains energy-concentrated under moderate channel variation, thereby supporting windowed DDA-domain sparse recovery. A local stability analysis further shows that the substantial overlap between adjacent windows enables the preceding-window estimate to warm-start each window-specific recovery problem, allowing the optimization progress to be carried across slots under a fixed per-slot iteration budget. For computational tractability, ST-DDA employs the alternating subspace method, which restricts the regularized least-squares fidelity updates to support-induced subspaces, together with position-encoded convolutional reweighting that exploits local angular and Doppler structures. Experiments show that reweighted ST-DDA achieves more accurate and reliable reconstruction than dynamic compressed-sensing baselines, particularly for longer sounding intervals and larger frequency-hopping periods, while maintaining comparable per-slot runtime.

\end{abstract}

\begin{IEEEkeywords}

Alternating subspace method, channel estimation, Doppler--delay--angle domain, frequency hopping, dynamic channel tracking, structured sparse recovery.
    
\end{IEEEkeywords}

\section{Introduction}\label{sec:intro}

\IEEEPARstart{S}{parse} recovery~\cite{CS,CSTao} has been widely used for channel estimation and denoising. In modern wireless communication systems, reliable high-throughput transmission relies on accurate channel acquisition. Owing to the compressibility of virtual-domain channel representations induced by sparse scattering paths, channel reconstruction can be performed with limited sounding resources. In TDD 5G NR orthogonal frequency-division multiplexing (OFDM) systems, sounding reference signals (SRSs) are transmitted in the uplink by mobile users and used by the base station (BS) to acquire user-specific channel state information (CSI).

Since the available sounding resources are limited, there is an inherent trade-off between user multiplexing and estimation accuracy. For a single sounding group in a TDD cell, supporting more users requires finer division of the pilot resources, which reduces the per-user pilot budget and makes channel recovery more challenging. In practical SRS-based TDD sounding specified in the 3GPP SRS-based frequency-hopping configuration~\cite{3gpp_ts38211_v1850}, this trade-off is further shaped by block-hopping allocation, under which the sounding subcarriers assigned to each user are clustered within a subband\cite{MCC}. Consequently, channel recovery becomes an increasingly severe frequency-extrapolation problem~\cite{candes} as the number of users in each group increases. As shown in Fig.~\ref{fig:pilot_patterns}, at the latest sounding slot $t_0$, nearby historical pilots may be more informative than the current-slot observation for an unobserved subband far from $\mathcal B_{t_0}$. Therefore, although many well-established compressed-sensing methods are available~\cite{4518398,berger2009sparse,5454399}, a single-slot architecture cannot fully exploit the slow channel variation commonly observed in practice.

To leverage historical information for latest-slot recovery, filtering-based~\cite{KalmanCS} dynamic compressed-sensing (CS) architectures retain a single-slot observation model while incorporating historical information as prior knowledge. The recent PLAY-CS method~\cite{PLAY-CS} exploits the slowly varying support and propagates the corresponding nonzero coefficients through a Kalman-$\ell_2$ regularization term while retaining an $\ell_1$ penalty outside the predicted support. Related support-prediction frameworks are developed in~\cite{RM-BPDN,6638908,ModifiedCS}. Message-passing methods such as~\cite{6557543} instead propagate extrinsic information along a hidden Markov chain that models support transitions and Gaussian-approximated coefficient priors. Similar message-passing architectures are adopted in~\cite{10311526,2Stage}. Nevertheless, their eventual reduction to a single-slot observation model limits their ability to exploit richer temporal correlations.

Alternatively, slow temporal variation can be characterized through a sparse representation in the Doppler domain. As shown in Fig.~\ref{fig:pilot_patterns}, at the latest sounding slot $t_0$, recent historical observations can be collected into a temporal estimation window and jointly recovered by exploiting DDA-domain sparsity. This DDA-sparse principle also underlies the well-known orthogonal time frequency space (OTFS) modulation~\cite{P7_OTFS,P8_raviteja2019embedded}, in which observations over an OTFS frame are integrated for DDA-domain sparse recovery by exploiting the temporal stationarity of massive MIMO channels. Recent tensor-based work~\cite{TensorBased} further demonstrates that DDA-domain modeling remains effective under a sliding-window framework for channel prediction in high-mobility OFDM systems with full-band comb observations. Once the DDA-domain channel is obtained, the latest-slot channel can be recovered through a simple time--Doppler transformation. Related DDA structures are also exploited via matrix-pencil estimation~\cite{matpencil} and triple-beam channel acquisition~\cite{HF_skywave} under full-band observations.

\begin{figure}[t]
\centering
\includegraphics[width=0.9\columnwidth]{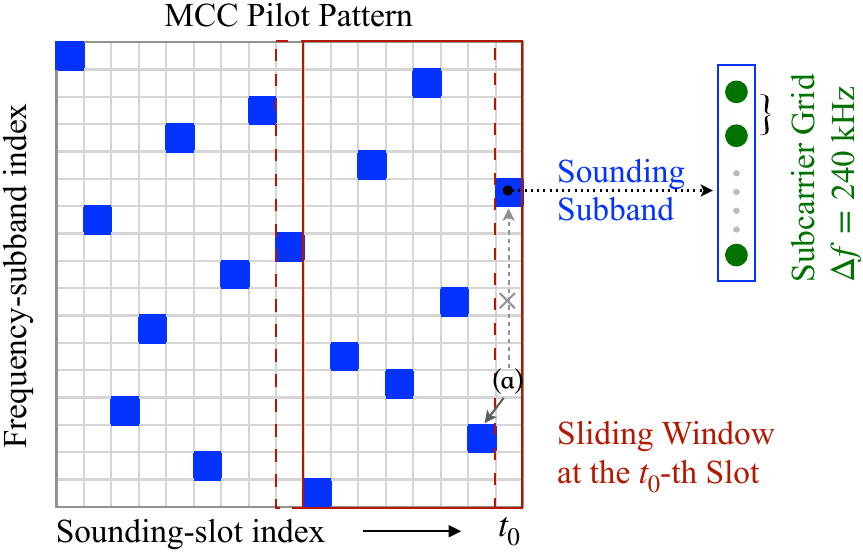}
\caption{Illustration of frequency-hopping SRS sounding and sliding-window multi-slot recovery. At the $t_0$-th slot, only one frequency subband is sounded on the subcarrier grid, while the recent observations are collected within a window of size $T$. For an unobserved subband such as (a), nearby historical pilots may be more informative than the current-slot observation. The dashed and solid windows represent two adjacent estimation windows.}
\label{fig:pilot_patterns}
\par\vspace{-2.0ex}
\end{figure}

In frequency-hopping TDD systems, however, the irregular observation pattern shown in Fig.~\ref{fig:pilot_patterns} severely complicates multi-slot joint recovery, rendering DDA-domain sparse reconstruction highly underdetermined. To assimilate information across subbands, the recent two-stage method in~\cite{2Stage} fuses the latest $n_{\rm hop}$ sounded subbands under a quasi-static path-parameter model over one hopping cycle, constructing a common cross-band representation in which each delay--angle path is associated with a path-specific Doppler offset. In contrast, the JCEP method in~\cite{JCEP} retains subband-specific DDA representations and couples them through shared DDA locations and statistical parameters. Since the Doppler-domain representation of each subband is recovered from repeated visits across multiple hopping cycles, its default configuration aggregates $K=10$ recent cycles with a short $0.5$-ms sounding interval. Under the lower-overhead sounding regime considered here, the same-subband revisit interval can become substantially longer; for example, $n_{\rm hop}=12$ with a $20$-ms sounding interval yields a $240$-ms revisit interval, making cycle-based temporal assimilation increasingly restrictive.

In this work, building on our companion works~\cite{MCC,dda-net}, we aggregate the recent observations within the latest window into a single DDA-sparse representation. Our analysis shows that moderate channel variation across sounding slots still permits an energy-concentrated Doppler-domain representation over the observation window, thereby preserving the underlying sparse-recovery structure. The resulting multi-slot joint recovery problem, however, can be computationally burdensome. We address this challenge from a dynamic tracking viewpoint. As shown in Fig.~\ref{fig:pilot_patterns}, two adjacent windows of size $T$ share $T-1$ overlapping slots and therefore admit closely related recovery problems. This motivates warm-starting the current-slot recovery with the previous-slot estimate and carrying the optimization progress across slots. A local stability analysis further supports this cross-window inheritance. Based on these, the contributions of the proposed ST-DDA framework are summarized as follows:

(i) \textbf{Dynamic multi-slot architecture for an extended operating regime:}
We develop an online multi-slot DDA recovery architecture in which each slot retains a window-specific recovery objective while inheriting the preceding-window estimate under a fixed iteration budget. By retaining an explicit DDA model, ST-DDA directly exploits temporal structures, particularly Doppler-domain priors, and better accommodates longer sounding intervals and larger frequency-hopping periods under synchronization imperfections, while maintaining a runtime comparable to that of dynamic CS.

(ii) \textbf{Support-induced computational efficiency:}
We formulate an alternating subspace method (ASM) solver~\cite{ASM} for the DDA-domain sparse recovery problem while preserving its tensor structure. As operator-splitting methods~\cite{LSCOA}, ASM and the alternating direction method of multipliers (ADMM)~\cite{ADMM:Boyd} share the main iteration structure, whereas ASM restricts the regularized least-squares (RLS) fidelity update to support-induced product subspaces while retaining full-space residual feedback. This enables tensor-structured processing of the RLS update through a Sylvester equation, reducing the computational cost and its sensitivity to the size of the outer estimation space, while the same splitting structure facilitates the incorporation of structured priors.

(iii) \textbf{Structured-prior-aided recovery:}
Following~\cite{NIPS2010_2d6cc4b2,TIPcrL,JunFang2015}, we introduce a position-encoded convolutional reweighting scheme that exploits local angular structures and the Doppler-domain energy concentration induced by channel variation. The resulting Doppler-domain prior is particularly beneficial at larger sounding intervals, while randomized neighborhood sampling limits the additional overhead.

The framework is organized modularly. Secs.~\ref{sub:multipath_channel_model} and~\ref{sub:Lasso_model} introduce the system model and fixed-window DDA recovery formulation, with the Doppler energy-concentration analysis given in Sec.~\ref{sub:group_sparse_representation}. Sec.~\ref{sec:Online_Tracking_Model} develops the online architecture and analyzes the warm-start stability across adjacent windows. The ASM solver and structured reweighting are presented in Secs.~\ref{sec:ASM_solver} and~\ref{sec:conv_reweigh}, followed by the history-aided estimation-space and synchronization modules in Secs.~\ref{sub:estimation_space} and~\ref{sub:non_ideal_MLE}. Numerical evaluations are provided in Sec.~\ref{sec:numerical_experiments}.

\section{System Model}

\subsection{Sounding Model: Hopping Pilots and Observations}

Consider a multi-user TDD massive MIMO-OFDM cell in which the users are organized into sounding groups. Each group comprises $n_{\rm hop}$ users and is allocated $N_{\mathrm f}$ full-band subcarriers, which are partitioned into $n_{\rm hop}$ nonoverlapping contiguous subbands of size $M_{\mathrm f}=N_{\mathrm f}/n_{\rm hop}$, as shown in Fig.~\ref{fig:pilot_patterns}. At each sounding slot, the subbands are assigned among the users, such that each user occupies one subband. The assignment follows a periodic block-hopping pattern with period $n_{\rm hop}$, as specified, e.g., by the 3GPP SRS configuration or the MCC-pilot in~\cite{3gpp_ts38211_v1850,MCC}. For a given user, let $\mathcal B_t$ denote the resulting active subcarrier set at the $t$-th sounding slot, where consecutive sounding slots are separated by $\delta t$.

The base station (BS) employs a uniform planar array (UPA) with $N_{\mathrm r}=N_x\times N_y\times N_p$ receive antennas, where $N_x$, $N_y$, and $N_p$ denote the horizontal, vertical, and polarization dimensions, respectively. Each user equipment (UE) is assumed to have a single transmit antenna. At the $t$-th sounding slot, the BS receives the SRS over the subband $\mathcal B_t$ and the resulting frequency--spatial observation $Y^{(t)}\in\mathbb C^{M_{\mathrm f}\times N_{\mathrm r}}$ satisfies
\begin{equation}
    Y^{(t)}
    =
    S^{(t)}\varPhi^{(t)}
    \xi(\ve^{(t)},\tau^{(t)})
    \tilde H^{(t)}
    +W^{(t)},
    \label{eq:obs}
\end{equation}
where $\tilde H^{(t)}\in\mathbb C^{N_{\mathrm f}\times N_{\mathrm r}}$ is the full-band frequency--spatial (FS) channel response, $\varPhi^{(t)}\in\{0,1\}^{M_{\mathrm f}\times N_{\mathrm f}}$ is the selection matrix extracting the rows of $\tilde H^{(t)}$ indexed by $\mathcal B_t$, and $W^{(t)}\in\mathbb C^{M_{\mathrm f}\times N_{\mathrm r}}$ is additive white Gaussian noise. The diagonal matrix $S^{(t)}$ contains the sounding sequence and satisfies $(S^{(t)})^HS^{(t)}=I$. The matrix $\xi(\ve^{(t)},\tau^{(t)})\in\mathbb C^{N_{\rm f}\times N_{\rm f}}$ models the residual synchronization perturbations, where $\ve^{(t)}$ and $\tau^{(t)}$ denote the phase offset and timing error, respectively. Specifically, $\xi(\ve,\tau)$ is diagonal, with its $n$-th diagonal element given by
$\xi_{nn}(\ve,\tau):=e^{\mathrm i\ve}e^{-2\pi\mathrm i\cdot n\tau/N_{\rm f}}$.
The residual phase and timing offsets are modeled incrementally~\cite{803501,salim2014channel} as:
\begin{equation}
\begin{aligned}
    \varepsilon^{(t)}
    &=
    \varepsilon^{(t-1)}+\dep^{(t)},
    \qquad
    \dep^{(t)}\in(-\pi,\pi],\\
    \tau^{(t)}
    &=
    \tau^{(t-1)}+\dtau^{(t)},
    \qquad
    \dtau^{(t)}\in[-\tau_0,\tau_0].
    \label{eq:NonIdeal}
\end{aligned}
\end{equation}
Without loss of generality, we set $S^{(t)}=I$ throughout this paper. To simplify the formulation of the DDA-sparse recovery model, we further assume $\xi(\ve^{(t)},\tau^{(t)})=I$ until Sec.~\ref{sub:streamline_for_online_processing}.

\subsection{Channel Model: Sparsity in the Virtual DDA-Domain}\label{sub:multipath_channel_model}

% We adopt the following sparse multipath channel model. At $t$-th sounding slot, the time--frequency--spatial (TFS) domain channel response $\tilde H^{(t)}$ at the $m$-th subcarrier and the $(i,j)$-th antenna can be decomposed into several paths in the virtual DDA domain such that the location of each path is determined by the delay $\tau_l^{(t)}$, the elevation angle of arrival (AoA) $\phi_x^{(t)}$ and the azimuth AoA $\phi_y^{(t)}$ in the $t$-th sounding slot. For the $l$-th path $(l\in\{1,2,\ldots,L\})$, we have
We adopt a sparse multipath channel model in the virtual DDA domain. At the $t$-th sounding slot, the frequency--spatial channel coefficient at frequency $f$ and antenna index $(i,j)$ is represented as a superposition of $L$ propagation paths. Each path is characterized by a delay, two angular coordinates, a Doppler-induced phase, and a complex gain. Specifically,
\begin{equation}
\begin{aligned}
    \tilde{H}_{f,i,j}^{(t)}=\sum_{l=1}^L &X^{(t)}_l e^{-2\pi\ii f\tau_l^{(t)}} e^{\ii\vp^{(t)}_l}  e^{-2\pi\ii  i\omega_{x,l}^{(t)}} e^{ -2\pi\ii  j  \omega_{y,l}^{(t)}},\label{eq:Model}
\end{aligned}
\end{equation}
where $\omega_{x,l}^{(t)}:=\frac{1}{2}\cos(\phi_{l}^{(t)})\cos(\psi_{l}^{(t)})$ and $\omega_{y,l}^{(t)}:=\frac{1}{2}\cos(\phi_{l}^{(t)})\sin(\psi_{l}^{(t)})$ are the angular-domain coordinates under half-wavelength antenna spacing \cite{booktse, bookUPA}. Here, $\phi_l^{(t)}$ and $\psi_l^{(t)}$ denote the elevation and azimuth angles, respectively.
For each path $l$, $\varphi_l^{(t)}$ denotes the accumulated Doppler-induced phase, whose slot-to-slot increment is determined by the Doppler shift $\nu_l^{(t)}$ according to $\varphi_l^{(t)}=\varphi_l^{(t-1)}+2\pi\delta t\,\nu_l^{(t)}$, where $\delta t$ denotes the sounding interval in seconds.
% For each path $l$, the Doppler shift $\vp^{(t)}_l$ interprets the channel's varying speed and is intrinsically controlled by the Doppler incremental $\nu_l^{(t)}$ satisfying $\varphi_l^{(t)}=\varphi_l^{(t-1)}+2\pi \delta t\cdot \nu_{l}^{(t)}$, where $\delta t$ indicates the sounding interval in seconds. 

\section{Multi-Slot Joint Channel Recovery over a Local Fixed Observation Window}

\subsection{Multi-Slot Joint Channel Recovery Model}\label{sub:Lasso_model}

We first formulate the sparse-recovery problem for the latest slot $t_0$ over a fixed observation window $\Omega_{t_0}:=\{t_0-T+1,\ldots,t_0\}$. As a simplified starting point, we assume that the effective channel variation is sufficiently slow such that the DDA parameters $\nu_l^{(t)}$, $\tau_l^{(t)}$, $\omega_{x,l}^{(t)}$, and $\omega_{y,l}^{(t)}$, together with the channel gain $X_l^{(t)}$, remain approximately constant for all $t\in\Omega_{t_0}$. Given the discretized DDA grid $\mathcal D_{t_0}:=(\mathcal T_{t_0},\varTheta_{t_0},\mathcal V_{t_0})$ for $\Omega_{t_0}$, the multipath channel model in Sec.~\ref{sub:multipath_channel_model} yields the following \textit{DDA channel estimation model} in LASSO formulation:
\begin{equation}
\begin{aligned}
    \min_{X,H}& \sum_{t\in\Omega_{t_0}}\frac1{2\sigma_t}\big\|Y_t-A_tH_{t}F_{\mathrm r}^\top\big\|_{\rm F}^2+\lambda\sum_{(\tau,\theta,\nu)\in\mathcal D_{t_0}}|X_{\tau,\theta,\nu}|,\\ 
    \text{s.t.}&\ H_{\tau,\theta,t}=\sum_{\nu\in \mathcal V_{t_0}}(F_{\mathrm d})_{t,\nu}X_{\tau,\theta,\nu}, \ \big(\forall \tau\in\mathcal T_{t_0}, \theta\in\varTheta_{t_0}\big).\label{eq:SingleWindowModel}
\end{aligned}
\end{equation}
Here, $H\in\mathbb C^{|\mathcal T_{t_0}|\times|\varTheta_{t_0}|\times T}$ denotes the delay--angular--time-domain channel representation, and $H_t\in\mathbb C^{|\mathcal T_{t_0}|\times|\varTheta_{t_0}|}$ is its slice at slot $t\in\Omega_{t_0}$, where the notation $|\mathcal{A}|$ denotes the cardinality of set $\mathcal{A}$. Moreover, $X\in\mathbb C^{|\mathcal T_{t_0}|\times|\varTheta_{t_0}|\times|\mathcal V_{t_0}|}$ denotes the delay--angular--Doppler-domain representation, and $F_{\mathrm d}\in\mathbb C^{T\times |\mathcal V_{t_0}|}$, $A_t\in \mathbb C^{M_{\mathrm f}\times|\mathcal T_{t_0}|}$, $F_{\mathrm r}\in\mathbb C^{N_{\mathrm r}\times|\varTheta_{t_0}|}$ are the time--Doppler, frequency--delay, and spatial--angular domain transformation matrices, respectively, to be precisely defined below. The introduction of the separate variables $H$ and $X$ facilitates the operator-splitting solver developed in Sec.~\ref{sec:ASM_solver}. Once $H_{t_0}$ is obtained, the reconstructed latest-slot full-band channel is given by $\tilde H^{(t_0)}=F_{\mathrm f}H_{t_0}F_{\mathrm r}^{\top}$, where $F_{\mathrm f}\in\mathbb C^{N_{\rm f}\times|\mathcal T_{t_0}|}$ is the frequency--delay super-resolution matrix.

Both super-resolution matrices $F_{\mathrm d}, F_{\mathrm f}$ are row-orthonormal partial discrete Fourier transform (DFT) matrices. The spatial--angular transformation matrix $F_{\mathrm r}$ is structured as $F_{\mathrm r}=\mathrm{diag}(\tilde F_{\mathrm r},\tilde F_{\mathrm r})$ for dual-polarized reception, where $\tilde F_{\mathrm r}$ is the single-polarization super-resolution matrix formed as the Kronecker product of the horizontal and vertical angular matrices, i.e., $\tilde F_{\mathrm r}=F_y\otimes F_x$. The frequency--delay sensing matrix $A_t$ has the form $A_t=S^{(t)}\varPhi^{(t)}\xi(\ve^{(t)},\tau^{(t)})F_{\mathrm f}$.

\subsection{Group Sparse Representation in the Doppler Domain}\label{sub:group_sparse_representation}

In practice, although the location parameters  typically vary only slightly over a moderate observation window~\cite{booktse,3GPP38901}, the path gain $X_l^{(t)}$ may fluctuate more rapidly, making the static-gain assumption restrictive. Nevertheless, we show that, even when $X_l^{(t)}$ and $\nu_l^{(t)}$ vary with time, the resulting Doppler-domain representation remains energy-concentrated and therefore admits a sparse approximation.
For a given path $l$, we assume that $\tau_l^{(t)}$, $\omega_{x,l}^{(t)}$, and $\omega_{y,l}^{(t)}$ remain constant for $t\in\Omega_{t_0}$. With a slight abuse of notation, we fix the path $l$, suppress the path index, write $x(t):=X_l^{(t)}$ and $\vp(t):=\vp_l^{(t)}$, and use a continuous-time notation over $t\in[0,T_{\rm w}]$ for this analysis. If we denote $y(t):=x(t)e^{\mathrm i\vp(t)}$ as the delay--angular--time-domain coefficient associated with the $l$-th path, then we have
\begin{prop}\label{prop:smooth_Doppler}
Suppose that $x(t)$ is complex-valued and continuously differentiable on $[0,T_{\rm w}]$, and that the phase $\vp(t)$ is twice continuously differentiable on $[0,T_{\rm w}]$, with $|x(t)|\leq c_0$, $|x'(t)|/c_0\leq r_1$, $\nu_{\min}\leq \vp'(t)/(2\pi)\leq \nu_{\max}$, and $|\vp''(t)|\leq r_2$ for $t\in[0,T_{\rm w}]$. Then, for any $\Delta \nu>0$ and $\nu\in\mathcal I_{\Delta \nu}^c$, where $\mathcal I_{\Delta \nu}:=\{\nu:\mathrm{dist}(\nu,[\nu_{\min}, \nu_{\max}])\leq\Delta \nu\}$, we have $\tilde x(\nu):=\frac1{T_{\rm w}}\int_{0}^{T_{\rm w}}y(t)e^{-2\mathrm i\pi \nu t}\mathrm dt$ satisfies:
\begin{equation}
|\tilde x(\nu)|\leq c_0\left(\frac{1}{T_{\rm w}\pi\Delta \nu} +\frac{r_1}{2\pi\Delta \nu}
+\frac{r_2}{(2\pi\Delta \nu)^2}\right).\label{eq:prop1}
\end{equation}
\end{prop}
\begin{IEEEproof}
See Appendix.~\ref{app:prop_smooth_doppler}.
\end{IEEEproof}

Here $\tilde x(\nu)$ is the Doppler representation under the rectangular observation window. The interval $[\nu_{\min},\nu_{\max}]$ describes the Doppler range swept over $t\in[0,T_{\rm w}]$, outside which the Doppler-domain coefficient $\tilde x(\nu)$ decays according to \eqref{eq:prop1}. The first term in~\eqref{eq:prop1} is induced by the rectangular observation window and represents the baseline finite-window spectral leakage. To characterize the structured sparsity induced by channel variation, the latter two terms are more relevant, as they quantify the additional Doppler spreading caused by the variations of the coefficient $x(t)$ and phase $\vp(t)$, respectively.

When $r_1$ and $r_2$ are small, this additional spreading remains limited, and the Doppler-domain representation stays concentrated around the swept Doppler interval. For a slowly moving user, the instantaneous Doppler shift $\vp'(t)/(2\pi)$ typically varies only mildly over a short window. Although $|\nu_{\max}-\nu_{\min}|$ and $r_2$ may increase for a longer observation window, the resulting representation can still exhibit an expanded yet concentrated Doppler support.

\section{Online DDA-Domain Recovery for Dynamic Channel Tracking}\label{sec:Online_Tracking_Model}

\subsection{Motivation: Local Stability Across Consecutive Recovery Problems}\label{sub:motivation_OLO}

In dynamic tracking, adjacent windows share the observations on $\Omega_{t_0-1}\cap\Omega_{t_0}$. This substantial overlap motivates an online optimization strategy that warm-starts the $\Omega_{t_0}$-based recovery problem with the sparse estimate obtained from $\Omega_{t_0-1}$.
To formalize this intuition without introducing the full channel notation, we consider an illustrative Doppler-domain subproblem for a fixed path with invariant delay--angular parameters, as motivated by Sec.~\ref{sub:group_sparse_representation}.

For a sequence $\{y_t\}_{t=1}^{T+1}$, define $\ell_t(u):=\frac12|y_t-(F_{\mathrm d})_tu|^2$ and $\mathcal U(\Omega):=\arg\min_u\sum_{t\in\Omega}\ell_t(u)+\lambda\|u\|_1$ for $u\in\mathbb C^{N\times 1}$. Let $\Omega_-:=\{1,\ldots,T\}$, $\Omega_+:=\{2,\ldots,T+1\}$, and let $\Omega_0:=\{1,\ldots,T+1\}$ denote the parent window. Denote the corresponding minimizers by $\hat x_-:=\mathcal U(\Omega_-)$, $\hat x_+:=\mathcal U(\Omega_+)$, and $\hat x_0:=\mathcal U(\Omega_0)$. Then the following result shows that the two adjacent-window solutions are close when the parent-window problem admits a stable optimizer $\hat x_0$:

\begin{prop}\label{prop:parent_child_error}
Suppose that there exists a set $\mathcal S\subseteq[N]$ such that $\operatorname{supp}(\hat x_0)\subseteq \mathcal S$. Define the parent KKT certificate $z_0:=-\lambda^{-1}\sum_{t=1}^{T+1}\nabla \ell_t(\hat x_0)$, the dual margin $\gamma:=1-\|z_0|_{{\mathcal S}^c}\|_\infty$, and the endpoint perturbations $\delta_-:=\|\nabla \ell_{T+1}(\hat x_0)\|_{\infty}$ and $\delta_+:=\|\nabla \ell_{1}(\hat x_0)\|_{\infty}$. If $\gamma\in(0,1)$, $\delta_-,\delta_+\leq\lambda\gamma/2$, and there exists $\kappa_c>0$ such that $\sum_{t=2}^T|(F_{\mathrm d})_th|^2\geq \kappa_c \|h\|^2_2$ holds for any $h$ satisfying $\|h|_{{\mathcal S}^c}\|_1\leq \|h|_{\mathcal S}\|_1$, then we have
\begin{equation}
    \|\hat x_+-\hat x_-\|_2\leq  4(\delta_++\delta_-)\kappa_c^{-1}\sqrt{|\mathcal S |}.\label{eq:parent_bound}
\end{equation}
\end{prop}
\begin{IEEEproof}
See Appendix.~\ref{app:prop_parent_child_error}.
\end{IEEEproof}

As investigated in Sec.~\ref{sub:group_sparse_representation}, when the Doppler-domain representation over a multi-slot window is generated by the time-varying ground truth $X_l^{(t)}$, its sparse-representation stability may deteriorate as the window size increases. Proposition~\ref{prop:parent_child_error} shows that, if the one-slot expansion to the parent-window problem still yields an optimizer $\hat x_0$ with a stable sparse screen, namely with small endpoint residual correlations $\delta_\pm$ and sufficient restricted curvature $\kappa_c$ over the shared observations $\Omega_-\cap\Omega_+$, then the two adjacent-window problems admit similar optimizers. Thus, the next-window problem can inherit an intermediate iterate from the previous window.

\subsection{The ST-DDA Architecture for Online Processing}\label{sub:streamline_for_online_processing}

\begin{figure}[t]
\centering
\includegraphics[width=0.9\columnwidth]{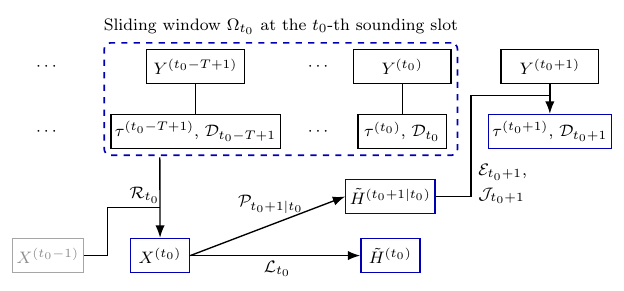}
\caption{Online ST-DDA processing over a sliding window. At the $t_0$-th sounding slot, the preceding-window estimate $X^{(t_0-1)}$ warm-starts the DDA-domain reconstruction $\mathcal R_{t_0}$ to obtain $X^{(t_0)}$, which is mapped by $\mathcal L_{t_0}$ to the latest-slot full-band channel $\tilde H^{(t_0)}$ and by $\mathcal P_{t_0+1|t_0}$ to the next-slot prediction $\tilde H^{(t_0+1|t_0)}$. Upon receiving $Y^{(t_0+1)}$, the prediction is compared with the new observation to estimate the residual phase and timing offsets $\varepsilon^{(t_0+1)}$ and $\tau^{(t_0+1)}$ through $\mathcal E_{t_0+1}$ and determine the estimation space $\mathcal D_{t_0+1}$ through $\mathcal J_{t_0+1}$, thereby enabling the next-window recovery.}
\label{fig:stream1}
\end{figure}

We now instantiate this cross-window inheritance principle in the online ST-DDA pipeline. At the latest sounding slot $t_0$, once the estimation space $\mathcal D_{t_0}$ has been determined and the residual phase and timing offsets $\{\tau^{(t)},\varepsilon^{(t)}\}_{t\in\Omega_{t_0}}$ have been estimated, the DDA-domain channel can be recovered by applying the joint model~\eqref{eq:SingleWindowModel} over $\Omega_{t_0}$.
As shown in Fig.~\ref{fig:stream1}, we denote this process of reconstructing the DDA-domain channel $X^{(t_0)}$ from the $\Omega_{t_0}$-windowed observation as follows
\begin{equation}
    X^{(t_0)}=\mathcal R_{t_0}\bigl(\{Y^{(t)}\}_{t\in\Omega_{t_0}},\{(\tau^{(t)},\varepsilon^{(t)})\}_{t\in\Omega_{t_0}},{\mathcal D}_{t_0}\bigr),\label{eq:DDASolver-t0}
\end{equation}
where $\mathcal R_{t_0}$ denotes any solver for \textit{DDA model}~\eqref{eq:SingleWindowModel}, such as the ASM solver introduced later in Sec.~\ref{sec:ASM_solver}.
As discussed in Sec.~\ref{sub:motivation_OLO}, $\mathcal R_{t_0}$ is not required to solve each window-specific problem to full convergence. Instead, it is warm-started by $X^{(t_0-1)}$, thereby carrying the optimization progress across sounding slots. From the resulting DDA-domain representation $X^{(t_0)}$, we obtain both the latest-slot full-band channel $\tilde H^{(t_0)}$ and the one-slot-ahead prediction $\tilde H^{(t_0+1|t_0)}$.

% For deriving the TFS-domain $\tilde H^{(t_0)}$, one can simply apply the tri-linear transform $\mathcal L_{t_0}$ induced by $F_{\mathrm d},F_{\mathrm f}$ and $F_{\mathrm r}$. As for $\tilde H^{(t_0+1|t_0)}$, we apply a similar tri-linear transform $\mathcal P_{t_0+1|t_0}$ which is different from $\mathcal L_{t_0}$ solely by substituting the time--Doppler transformation with $(F_{\mathrm d})_{t_0+1}$. As will be demonstrated in Sec.~\ref{sub:estimation_space} and Sec.~\ref{sub:non_ideal_MLE}, the synchronization of $\tau^{(t_0+1)},\varepsilon^{(t_0+1)}$ and the determination of $\mathcal D_{t_0+1}$ can be achieved under the assistance of $\tilde H^{(t_0+1|t_0)}$ by comparing it with the newly-coming $Y^{(t_0+1)}$. In Fig.~\ref{fig:stream1}, we denote the synchronization and estimation space determination operation as $\mathcal E_{t_0+1}$ and $\mathcal J_{t_0+1}$, respectively, and this closes the tracking process under window $\Omega_{t_0}$. 

The latest-slot full-band channel $\tilde H^{(t_0)}$ is obtained by applying the trilinear transform $\mathcal L_{t_0}$ induced by $F_{\mathrm d}$, $F_{\mathrm f}$, and $F_{\mathrm r}$. The prediction $\tilde H^{(t_0+1|t_0)}$ is obtained through an analogous transform $\mathcal P_{t_0+1|t_0}$, in which the time--Doppler factor is replaced by $(F_{\mathrm d})_{t_0+1}$. As will be demonstrated in Secs.~\ref{sub:estimation_space} and~\ref{sub:non_ideal_MLE}, the predicted channel is going to be compared with the newly received observation $Y^{(t_0+1)}$ to synchronize $\tau^{(t_0+1)}$ and $\varepsilon^{(t_0+1)}$ and to construct $\mathcal D_{t_0+1}$. These two operations are denoted by $\mathcal E_{t_0+1}$ and $\mathcal J_{t_0+1}$, respectively, in Fig.~\ref{fig:stream1}, thereby completing the processing associated with window $\Omega_{t_0}$.
Because a common DDA grid is used within each window, all observations in $\Omega_{t_0}$ are represented over the latest estimation space $\mathcal D_{t_0}$. During the warm-up stage $t_0<T$, we retain the same processing pipeline but use the truncated window $\Omega'_{t_0}:=\{1,\ldots,t_0\}$.

\section{Alternating Subspace Method for the Structured Sparse Recovery in the DDA Domain}\label{sec:ASM_solver}

\subsection{The Alternating Subspace Iteration Framework}\label{sec:ASM_Iteration}

Although the online architecture above avoids solving each windowed problem to full convergence, its practical efficiency still depends on the convergence speed under a fixed per-slot iteration budget, which directly affects the steady-state accuracy. The single-window DDA model~\eqref{eq:SingleWindowModel} can be viewed as a tensorized LASSO problem. 
Rather than vectorizing the problem and applying an unstructured solver, we preserve the separate variables $H$ and $X$ and adopt an operator-splitting formulation~\cite{LSCOA} under the constraint $H=F_{\mathrm d}X$. Following ASM~\cite{ASM}, an ADMM-type framework~\cite{ADMM:Boyd} tailored to sparse recovery through support-induced subspace updates, the $k$-th iteration for \textit{DDA model}~\eqref{eq:SingleWindowModel} is given by
\begin{align}
&X^{k,1}_{\tau,\theta}=\mathop{\arg\min}_{X_{\tau,\theta}} \mathcal G_\lambda(X_{\tau,\theta}; H^k_{\tau,\theta}+v\Pi^k_{\tau,\theta}), (\forall \tau,\theta),\label{line:ST}\\
    &H_{\tau,\theta}^{k,0}=F_{\mathrm d}X_{\tau,\theta}^{k,1}, (\forall \tau,\theta), \\
    &H_t^{k,1}|_{E^k}=\mathop{\arg\min}_{H_t|_{E^k}}\mathcal F_t(H_t|_{E^k};[H_t^{k,0}-v\Pi^k_t]|_{E^k}), (\forall t),\label{line:LMMSE1}\\
    &\Pi^{k,1}_t=\frac1{\sigma_t} A_t^H(Y_t- A_t H_t^{k,1} F_{\mathrm r}^\top )\bar{F_{\mathrm r}},\ (\forall t),\label{line:GD2}\\
    &\boldsymbol \Lambda^{k+1}=(1-d)\boldsymbol \Lambda^{k} + d\cdot \boldsymbol \Lambda^{k,1}, \label{line:AVE}
\end{align}
where $\boldsymbol\Lambda^k:=(X^k,H^k,\Pi^k)$ is the product-space iterate and $\boldsymbol\Lambda^{k,1}:=(X^{k,1},H^{k,1},\Pi^{k,1})$ collects the intermediate updates. Here, $v>0$ is the step size, $d\in(0,1)$ is the averaging factor, $t\in\Omega_{t_0}$, $\tau\in\mathcal T_{t_0}$, and $\theta\in\varTheta_{t_0}$. After $k_{\max}$ iterations, the slot estimate is set as $X^{(t_0)}:=X^{k_{\max}}$.

For inputs $Z\in\mathbb C^{|E^k_{\mathcal T_{t_0}}|\times|E^k_{\varTheta_{t_0}}|}$ and $Z'\in\mathbb C^{T\times1}$, the two subproblem objectives are
\begin{align}
    &\mathcal F_t(\hat H_t;Z):=\frac v{\sigma_t}\|Y_t-\hat A_t^k \hat H_t (\hat F_{\mathrm r}^k)^\top\|_{\rm F}^2+\|\hat H_t-Z\|_{\rm F}^2,\label{eq:RLS}\\
    &\mathcal G_\lambda(X_{\tau,\theta};Z'):=\frac1{2v}\|Z'-F_{\mathrm d}X_{\tau,\theta}\|_{\rm F}^2+\lambda\|X_{\tau,\theta}\|_1.\label{eq:CS-DS}
\end{align}
Here, $\hat H_t:=H_t|_{E^k}$ is optimized over the delay--angular subspace $E^k:=E^k_{\mathcal T_{t_0}}\times E^k_{\varTheta_{t_0}}$, which is determined from the support of $X^{k,1}$ induced by the $\mathcal G_\lambda$-subproblem; its construction is detailed in Sec.~\ref{sub:f_t_subproblem}. Correspondingly, $\hat A_t^k$ and $\hat F_{\mathrm r}^k$ are the column-restricted submatrices of $A_t$ and $F_{\mathrm r}$ associated with $E^k$. The complementary entries are left unchanged as $H_t^{k,1}|_{(E^k)^c}:=H_t^{k,0}|_{(E^k)^c}$, which remain zero by construction.

The subspace restriction in~\eqref{line:LMMSE1} and the averaging step~\eqref{line:AVE} follow the core principle of ASM~\cite{ASM}: the expensive data-fidelity update in sparse recovery is restricted to a support-induced subspace for efficiency, while a full-space gradient correction preserves residual feedback outside the current subspace and averaging stabilizes the iteration against imperfect support identification. In the present formulation, choosing $E^k$ as the full delay--angular estimation space $\mathcal T_{t_0}\times\varTheta_{t_0}$ and setting $d=1$ recover the corresponding full-space ADMM iteration under the splitting $H=F_{\mathrm d}X$. We instead restrict the $\mathcal F_t$-subproblem to $E^k$, retain the full-space correction~\eqref{line:GD2}, and adopt $d\in(0,1)$ for stabilization.

In practice, neither subproblem needs to be solved to high precision at every iteration. The employed updates for the $\mathcal G_\lambda$- and $\mathcal F_t$-subproblems are introduced in Secs.~\ref{sub:g_lambda_subproblem} and~\ref{sub:f_t_subproblem}, respectively. To prepare for the convolutional reweighting in Sec.~\ref{sec:conv_reweigh}, we further generalize the scalar $\lambda$ to position- and iteration-dependent weights $\lambda_\zeta^k$, where $\zeta:=(\tau,\theta,\nu)\in\mathcal D_{t_0}$. These weights are treated as given throughout this section, with their construction deferred to Sec.~\ref{sec:conv_reweigh}.

\subsection{Doppler Solver for the $\mathcal G_\lambda$-Subproblem}\label{sub:g_lambda_subproblem}

For given weights $\{\lambda_{\zeta}^k\}$, the $\mathcal G_\lambda$-subproblem is a LASSO problem arising from the time--Doppler duality and thus requires an iterative solution. A simple choice is to perform one iterative shrinkage-thresholding algorithm (ISTA)~\cite{ISTA1} step within each main-loop iteration. Let $U^k_{\tau,\theta}:=H^k_{\tau,\theta}+v\Pi^k_{\tau,\theta}\in\mathbb C^{N_{\mathrm d}\times 1}$ denote the input to $\mathcal G_\lambda$ in~\eqref{line:ST}, define $\tilde {\boldsymbol{\lambda}}^k_{\tau,\theta}\in\mathbb R^{N_{\rm d}\times 1}$ with entries $\{v\gamma \lambda^k_{\tau,\theta,\nu}\}_{\nu\in\mathcal V_{t_0}}$ and let $X^{k,0}:=X^k, X^{k+1}:=X^{k,1}$. The one-step ISTA update is
\begin{equation}
    X^{k,1}_{\tau,\theta}={\mathrm{ST}}_{\tilde {\boldsymbol{\lambda}}_{\tau,\theta}^{k+1}}[X^{k,0}_{\tau,\theta}+\gamma F_{\mathrm d}^H(U^k_{\tau,\theta}-F_{\mathrm d}X^{k,0}_{\tau,\theta})],\label{eq:sub_ISTA}
\end{equation}
where $\mathrm{ST}_{\boldsymbol \beta}(x):=\mathrm{sgn}(x) \odot \max\{|x| - \boldsymbol \beta, 0\}$ denotes the componentwise soft-thresholding operator. Here, $\gamma$ is the ISTA step size, and $\mathrm{sgn}(x):= x/|x|$ is understood componentwise for nonzero entries, with $\mathrm{sgn}(0):=0$.

For relatively large sounding intervals, the Doppler-domain channel energy may exhibit a more complicated structure, motivating a stronger inner update. Specifically, inserting the following additional data-fidelity step at each $(\tau,\theta)$ grid point before~\eqref{eq:sub_ISTA} yields an ADMM iteration:
\begin{equation}
    X^{k,0}_{\tau,\theta}= (I+\gamma F_{\mathrm d}^HF_{\mathrm d})^{-1}(X^{k}_{\tau,\theta}+\gamma F_{\mathrm d}^H U^k_{\tau,\theta} - \tilde {\boldsymbol \lambda}^{k}_{\tau,\theta} \odot \mathrm{sgn}(X^{k}_{\tau,\theta})).\label{eq:sub_data_fidelity}
\end{equation}
Since the Doppler-grid size $|\mathcal V_{t_0}|$ is small, the matrix inversion in~\eqref{eq:sub_data_fidelity} is inexpensive. Once the inverse is available, however, the associated matrix--vector multiplication must still be repeated over the full delay--angular grid of size $|\mathcal T_{t_0}|\times|\varTheta_{t_0}|$. To reduce this cost, we again impose an ASM-type subspace restriction and selectively apply~\eqref{eq:sub_data_fidelity} to a delay--angular subspace determined by the support of $X^k$.

When $\lambda$ is a fixed scalar, the ordering of the two updates above is less consequential. Once the reweighting of $\{\lambda_{\zeta}^k\}_{\zeta\in\mathcal D_{t_0}}$ is introduced, however, it is preferable to update the weights after the data-fidelity step, as indicated in~\cite{TIPcrL}. Accordingly, the ASM-type update follows
\begin{equation}
    X^k\rightarrow X^{k,0}|_{\tilde E^k}\rightarrow\{\lambda_{\zeta}^{k+1}\}_{\zeta\in\mathcal D_{t_0}}\rightarrow X^{k,1}:=X^{k+1}.\label{eq:order_update}
\end{equation}
Here, $\tilde E^k$ denotes the delay--angular subspace on which $X^{k,0}$ is updated according to~\eqref{eq:sub_data_fidelity}, whereas $X^{k,0}_{\tau,\theta}=X^k_{\tau,\theta}$ is retained for $(\tau,\theta)\in (\tilde E^k)^c$. Although both $\tilde E^k$ and $E^k$, introduced for the $\mathcal F_t$-subproblem in~\eqref{line:LMMSE1}, are delay--angular subspaces, they are constructed differently. As shown in Sec.~\ref{sub:f_t_subproblem}, $E^k:=E_{\mathcal T_{t_0}}^k\times E_{\varTheta_{t_0}}^k$ must be a product space so that the restricted variable $H_t|_{E^k}$ admits a matrix representation. In contrast, the $\mathcal G_\lambda$-subproblem is separable over $(\tau,\theta)$, allowing $\tilde E^k$ to specify a finer subspace.
In implementation, $X^k$ undergoes the averaging step~\eqref{line:AVE} before entering the $\mathcal G_\lambda$-subproblem. Consequently, some previously active $(\tau,\theta)$ nodes may remain numerically nonzero and inflate the support under a strict nonzero criterion. We therefore introduce a small threshold $\epsilon_1$ (e.g., $\epsilon_1=10^{-8}$) and define the practical subspace $\tilde E^k:=\{(\tau,\theta)\in\mathcal T_{t_0}\times\varTheta_{t_0} :\|X^k_{\tau,\theta}\|_\infty\geq \epsilon_1 \}$.

\begin{algorithm}[t]
\caption{ASM Solver for $\mathcal R_{t_0}$ in~\eqref{eq:DDASolver-t0} with Reweighting}
\label{alg:ASM_solver}
\begin{algorithmic}[1]
\REQUIRE $\{Y_t, \sigma_{t}, \varepsilon^{(t)}, \tau^{(t)}, \mathcal B_{t} \}_{t\in\Omega_{t_0}}$, $\mathcal D_{t_0}:=(\mathcal T_{t_0},\varTheta_{t_0},\mathcal V_{t_0})$; Initialize all $\lambda_{\zeta}$ as $\lambda_{t_0}$, $\boldsymbol\Lambda^0=(X^0,H^0,\boldsymbol 0)$.
\ENSURE $X^{(t_0)}$ and $X^0,H^0$ for $(t_0+1)$-th slot
\FOR{$k=0,\ldots,k_{\max}-1$}
    \STATE \% $\mathcal G_\lambda$-\textbf{subproblem} \eqref{line:ST}
      (for all $(\tau,\theta)\in\mathcal T_{t_0}\times\varTheta_{t_0}$)
    \STATE Derive the input $U^k_{\tau,\theta}=H^k_{\tau,\theta}+v\Pi^k_{\tau,\theta}$; 
    \STATE Construct $\tilde E^k:=\{(\tau,\theta)\in\mathcal T_{t_0}\times\varTheta_{t_0} :\|X^k_{\tau,\theta}\|_\infty\geq \epsilon_1 \}$ 
    \STATE Derive $X^{k,0}|_{\tilde E^k}$ from $U^k$ using~\eqref{eq:sub_data_fidelity}.
    \STATE Retain $X^{k,0}|_{(\tilde E^k)^c}=X^k|_{(\tilde E^k)^c}$.
    \STATE Reweight $\lambda_\zeta^{k+1}=\lambda_{t_0}/(|\mathcal N_\zeta|^{-1}\sum_{\zeta'\in\mathcal N_\zeta}|X^{k,0}_{\zeta'}|+\epsilon)$.
    \STATE Derive $X^{k,1}_{\tau,\theta}$ from $X^{k,0}$ via ISTA in~\eqref{eq:sub_ISTA}.
    \STATE Time--Doppler Transition: $H_{\tau,\theta}^{k,0}=F_{\mathrm d}X_{\tau,\theta}^{k,1}$.
    \STATE \% $\mathcal F_t$-\textbf{subproblem} \eqref{line:LMMSE1}
    \STATE Construct
    $E^k=E_{\mathcal T_{t_0}}^k\times E_{\varTheta_{t_0}}^k$
    from the support of $X^{k,1}$
    \FOR{$t\in\Omega_{t_0}\backslash\tilde\Omega^k_{t_0}$}
        \STATE Low-Rank Updates: $(I+\alpha P^k)^{-1}$, $(I+\alpha Q^k)^{-1}$
        \STATE Derive $H_t^{k,1}|_{E^k}$ by~\eqref{eq:syl_scheme} from $H_t^{k,0}|_{E^k}$.
        \STATE Retain $H_t^{k,1}|_{(E^k)^c}=H_t^{k,0}|_{(E^k)^c}$.
        \STATE $\Pi^{k,1}_t=\sigma_t^{-1} A_t^H(Y_t- A_t H_t^{k,1} F_{\mathrm r}^\top )\bar{F_{\mathrm r}}$.
    \ENDFOR
    \STATE Set $(H^{k,1}_t,\Pi^{k,1}_t)=(H^{k}_t,\Pi^{k}_t)$ for $t\in\tilde\Omega^{k}_{t_0}$.
    \STATE Set
    $\boldsymbol\Lambda^{k+1}
    =(1-d)\boldsymbol\Lambda^k+d\boldsymbol\Lambda^{k,1}$ by~\eqref{line:AVE}.
\ENDFOR
\STATE Set $X^{(t_0)}=X^{k_{\max}}$, and synchronize for $(t_0+1)$-th slot: 
\STATE $X^{0}_{\tau,\theta} \leftarrow \sqrt{N_{\rm d}}(F_{\mathrm d})_{2}^\top \odot X^{(t_0)}_{\tau,\theta},\ H^{0}_{\tau,\theta,t_0+1}\leftarrow (F_{\mathrm d})_{T}X^{0}_{\tau,\theta}$
\end{algorithmic}
\end{algorithm}

\subsection{Structured Solution of the $\mathcal F_t$-Subproblem}\label{sub:f_t_subproblem}

The subspace restriction in~\eqref{line:LMMSE1} is introduced to reduce the cost of obtaining the RLS update in~\eqref{eq:RLS}. However, constructing a vectorized observation dictionary for the active delay--angular pairs $(\tau,\theta)$ would destroy the row--column tensor structure of the subproblem. To preserve this structure and avoid solving a large vectorized linear system, we choose the restricted space in product form, i.e., $E^k:=E_{\mathcal T_{t_0}}^k\times E_{\varTheta_{t_0}}^k$, where $E_{\mathcal T_{t_0}}^k:=\{\tau\in\mathcal T_{t_0}: \|X^{k,1}_{\tau}\|_{\rm F}> 0 \}$ and $E_{\varTheta_{t_0}}^k:=\{\theta\in\varTheta_{t_0}: \|X^{k,1}_{\theta}\|_{\rm F}> 0 \}$. Under the product restriction, the $\mathcal F_t$-subproblem has the standard regularized matrix least-squares form, whose normal equation gives the following Sylvester equation~\cite{Benner2009ADI}:
\begin{equation}
    \mathop{\mathrm{Find}}_{\hat H_t}\ \hat H_t+ \frac{v}{\sigma_t} P^k\hat H_t(Q^k)^\top=C_t^k,\ (\forall t)\label{eq:Sylvester}
\end{equation}
where $\hat H_t\in\mathbb C^{|E^k_{\mathcal T_{t_0}}|\times |E^k_{\varTheta_{t_0}}|}$ is optimized over the product space $E^k$, with its rows and columns indexed by $E_{\mathcal T_{t_0}}^k$ and $E_{\varTheta_{t_0}}^k$, respectively. The coefficient matrices are given by $P^k:=(\hat A_t^k)^H\hat A_t^k$ with $\hat A_t^k:=A_t|_{E^k_{\mathcal T_{t_0}}}$, and $Q^k:=(\hat F_{\mathrm r}^k)^H\hat F_{\mathrm r}^k$ with $\hat F_{\mathrm r}^k:=F_{\mathrm r}|_{E^k_{\varTheta_{t_0}}}$. Moreover, $C_t^k:=[H_t^{k,0}-v\Pi^k_t+({v}/{\sigma_t}) A_t^HY_t(F_{\mathrm r}^H)^\top]|_{E^k}$, where the restriction is applied according to the same row--column index sets.

As a simple illustration, when one restricted sensing factor is semi-unitary, e.g., $(\hat F_{\mathrm r}^k)^H\hat F_{\mathrm r}^k=I$, the Sylvester equation reduces to linear solves over a single domain and only the inverse $(I+({v}/{\sigma_t})P^k)^{-1}$ is required. Let $m=|E^k_{\mathcal T_{t_0}}|$ and $n=|E^k_{\varTheta_{t_0}}|$. The cost of forming this inverse is $\mathcal O(m^3)$, while applying it to $C_t^k\in\mathbb C^{m\times n}$ costs $\mathcal O(m^2n)$. Hence, in the considered sparse regime, these small-domain inversions are not necessarily the dominant cost.

In the general case where both delay and angular grids are restricted to a sparse subspace, such a single-domain reduction is unavailable. Since the $\mathcal F_t$-subproblem is embedded in the outer ASM iteration, we do not solve the resulting generalized Sylvester equation to full convergence at each iteration. Instead, we keep the inversions decoupled over the two restricted domains and, motivated by a factorized alternating-direction implicit iteration~\cite{ADI} and splitting methods for linear matrix equations~\cite{Simoncini2016MatrixEquations,Benner2009ADI}, apply one inexact, structure-preserving splitting step. Specifically, with $\alpha:=\sqrt{v/\sigma_t}$, we update
\begin{equation}
\begin{aligned}
    H_t^{k,1}|_{E^k}\leftarrow& (I+{\alpha} P^k)^{-1}[C_t^{k}+{\alpha} P^k(H_t^{k,0}|_{E^k})\\
    &+{\alpha} (H_t^{k,0}|_{E^k})(Q^k)^\top](I+{\alpha} (Q^k)^\top)^{-1},\label{eq:syl_scheme}
\end{aligned}
\end{equation}
which can be viewed as a symmetric splitting step for~\eqref{eq:Sylvester}. The fixed point of~\eqref{eq:syl_scheme} coincides with the solution of the Sylvester equation, and the associated error iteration has spectral radius strictly smaller than one, as obtained by diagonalizing the Hermitian positive semidefinite factors $P^k$ and $(Q^k)^\top$. This gives a basic fixed-point consistency and contraction guarantee for the inexact $\mathcal F_t$-update. The two inverses in~\eqref{eq:syl_scheme} are formed only over the restricted delay and angular dimensions and are therefore amenable to the low-rank updates in~\cite{ASM}, since both $P^k$ and $Q^k$ are Gram matrices.

\subsection{Practical Strategies for Online Processing}\label{sub:online_processing}

Several additional strategies can further reduce the computational burden of online tracking. First, the joint recovery framework requires solving the $\mathcal F_t$-subproblem~\eqref{eq:RLS} for every $t\in\Omega_{t_0}$ at each iteration. Even under a fixed per-slot budget, allocating equal computational effort to recent and older observations is unnecessary because the recent observations are generally more relevant to the latest-slot reconstruction. Motivated by stochastic optimization methods for linear inverse problems~\cite{pmlr-v32-zhong14,siam:StochasticPDHG}, we therefore always update the recent-slot subproblems while randomly updating only a subset of the older-slot subproblems at each iteration. 

Specifically, at the $t_0$-th sounding slot, we partition the window into a recent part $\Omega^+_{t_0}:=\{t_0,t_0-1,\ldots,t_0-c_T+1\}$ and a historical part $\Omega^-_{t_0}:=\Omega_{t_0}\backslash\Omega^+_{t_0}$. At each iteration, $H_t^k$ is updated for every $t\in\Omega^+_{t_0}$ and for a randomly selected subset of $r_T$ slots from $\Omega^-_{t_0}$. The remaining historical slots, collected in $\tilde\Omega^k_{t_0}$, are frozen during that iteration.

% \textcolor{red}{
% As an online tracking task, there are additional strategies for further relieving the computational burden. First, the joint recovery framework involves solving the $\mathcal F_t$-subproblem~\eqref{eq:RLS} for all $t\in\Omega_{t_0}$ at every iteration. Even if we control the per-slot budget, equally assigning the computational resource can be unnecessary. Motivated by the stochastic optimization methods for linear inverse problem~\cite{pmlr-v32-zhong14,siam:StochasticPDHG}, we introduce a random dropout strategy by keeping the relevant lately comming slots $\mathcal F_t$-subproblem at each iteration while randomly freezing the long-term updates, leading to budget redistribution.
% Specifically, at the $t_0$-th sounding slot, we divide the window into a recent part $\Omega^+_{t_0}:=\{t_0,t_0-1,\ldots,t_0-c_T+1\}$ and a long-term historical part $\Omega^-_{t_0}:=\Omega_{t_0}\backslash \Omega^+_{t_0}$. At each iteration, we always update $H_t^k$ for $t\in\Omega^+_{t_0}$, while only updating $H_t^k$ for a randomly selected subset of $r_T$ slots from $\Omega^-_{t_0}$ and frozen on the the residual set $\tilde\Omega^{k}_{t_0}$.
% }

Second, we employ a support-set-based adaptation of $\lambda_{t_0}$, which remains fixed throughout the iterations at the $t_0$-th slot. After obtaining $X^{(t_0)}$, its final support is used to determine the regularization coefficient $\lambda_{t_0+1}$. 
We use $E_{\mathcal T_{t_0}}^{k_{\max}}$ as the reference set because the computational cost is particularly sensitive to the delay-grid size. Specifically, given a reference support size $C$ and an adjustment interval $C_r$ (e.g., $C=100$ and $C_r=50$), we first set
$\lambda_{t_0+1}:=\lambda_{t_0}\bigl[1+0.1\lfloor(|E_{\mathcal T_{t_0}}^{k_{\max}}|-C)/C_r\rfloor\bigr]$
and then clip $\lambda_{t_0+1}$ to the interval $[0.1\lambda_0,10\lambda_0]$.

\section{Convolutional Reweighting for DDA-Domain Structured Sparsity}\label{sec:conv_reweigh}

\subsection{Position Encoded Convolutional Reweighting}\label{sub:convolutional_reweighting}

The preceding section developed the ASM solver for the windowed model~\eqref{eq:SingleWindowModel} while treating the weights $\{\lambda_\zeta^k\}_{\zeta\in\mathcal D_{t_0}}$ as given. Classical coefficientwise reweighting can be interpreted through the Laplacian scale-mixture framework~\cite{NIPS2010_2d6cc4b2}: assigning each $\lambda_\zeta$ a $\mathrm{Gamma}(a,b)$ hyperprior $(a,b>0)$ yields the posterior-mean update $\lambda_\zeta^{k+1}:=(a+1)/(b+|X_\zeta^k|)$, followed by solving the corresponding fixed-weight LASSO problem. This inverse-amplitude mechanism assigns larger penalties to small estimated coefficients, thereby reinforcing sparsity and reducing the regularization bias on significant coefficients.

As shown in~\cite{NIPS2010_2d6cc4b2,TIPcrL,JunFang2015}, this mechanism can be generalized from individual coefficients to locally aggregated amplitudes, thereby incorporating position-encoded priors. Motivated by the Doppler-domain grouping characterized in Proposition~\ref{prop:smooth_Doppler} and the angular-domain grouping in~\cite{GSAng,STCS}, we adopt
\begin{equation}
    \lambda_{\zeta}^{k+1}:=\lambda_{t_0}\left/\left(\frac1{|\mathcal N_\zeta|}\sum_{\zeta'\in\mathcal N_\zeta}|X^k_{\zeta'}|+\epsilon\right)\right.,\label{eq:CLR_update}
\end{equation}
where $\lambda_{t_0}$ is a benchmark parameter, $\epsilon>0$ avoids singularity (e.g., $\epsilon=0.1$), and $\mathcal N_\zeta$ denotes the DDA-domain neighborhood acting as a convolutional kernel. Analogous to the elementwise inverse-amplitude mechanism, this rule favors coefficients supported by energetic neighborhoods while suppressing isolated coefficients, which are more likely to be artifacts caused by the high mutual coherence of the sensing matrix under the undersampled hopping pattern. Similar to the strategy in~\cite{TIPcrL}, where each reweighting update is followed by a single inner-solver iteration, we embed \eqref{eq:CLR_update} into the $\mathcal G_\lambda$-subproblem in~\eqref{eq:order_update}, with $X^k$ in~\eqref{eq:CLR_update} replaced by $X^{k,0}$ in~\eqref{eq:order_update}.

\subsection{The Randomized Convolutional Reweighting}

To reduce the additional computational cost of convolutional reweighting, we adopt the following strategies.
First, rather than using a product-space neighborhood in~\eqref{eq:CLR_update}, we decouple the aggregation over the Doppler and angular domains. For each central point $\zeta:=(\tau,\theta,\nu)\in\mathcal D_{t_0}$, we define an angular neighborhood $\mathcal N^{\varTheta}_\zeta$ and a Doppler neighborhood $\mathcal N^{\mathcal V}_\zeta$, and set $\mathcal N_\zeta:=\mathcal N^{\varTheta}_\zeta\cup \mathcal N^{\mathcal V}_\zeta$. Specifically, let $\bar\varTheta$ denote the full angular grid, to be detailed in Sec.~\ref{sub:estimation_space}, and decompose it as $\bar\varTheta:=\bar\varTheta_x\times\bar\varTheta_y\times\bar\varTheta_p$ over the horizontal, vertical, and polarization dimensions. For $\theta=(\theta_x,\theta_y,p)$, define $\boldsymbol r:=(r_x,r_y,r_p)$, together with the angular and Doppler bounds $r_{\rm ang}^c$ and $r_{\rm Dop}^c$, respectively, and set
\begin{align}
    \mathcal N^{\varTheta}_\zeta
    &:=\{(\tau,\theta+\boldsymbol r,\nu): |r_x|\le 1,\ |r_x|+|r_y|\le r_{\rm ang}^c\},\label{eq:neighbor_theta}\\
    \mathcal N^{\mathcal V}_\zeta
    &:=\{(\tau,\theta,\nu+r): |r|\le r_{\rm Dop}^c,\ \nu+r\in\mathcal V_{t_0}\}.\label{eq:neighbor_Dop}
\end{align}
We constrain $|r_x|\le 1$ because the horizontal angular dimension is typically small, e.g., $N_x=4$. This decoupled stencil preserves local energy aggregation while keeping the reweighting step lightweight.

Second, rather than averaging over the entire $\mathcal N_\zeta$ at every iteration, we randomly sample a subset $\widetilde{\mathcal N}_\zeta\subseteq \mathcal N_\zeta$ and evaluate~\eqref{eq:CLR_update} over $\widetilde{\mathcal N}_\zeta$. Since the sampled stencil is refreshed during the iterations, the reweighting rule can still explore a broader neighborhood over time.
Finally, since the reweighted parameters $\lambda_\zeta^k$ serve only as adaptive regularization coefficients, moderate numerical accuracy is generally sufficient. We therefore perform the entire reweighting procedure in single precision to reduce memory traffic, while leaving the main sparse-recovery iteration unchanged.

\section{Determination of the Estimation Space and Residual Synchronization Offsets}

\subsection{History-Aided Construction of the Estimation-Space}\label{sub:estimation_space}

The estimation-space construction effectively produces a nonuniform grid with locally varying resolution. A coarser grid reduces the computational cost but may fail to capture accurate path locations. In dynamic channel tracking, historical estimates can be used to refine the grid selectively, thereby balancing reconstruction accuracy and computational efficiency.
Using the delay domain as an example, we construct $\mathcal T_{t_0}:=\mathcal T_{0}\cup \mathcal T'_{t_0}\cup \mathcal T''_{t_0|t_0-1}$ from three components, serving as $\mathcal J_{t_0}$ in Sec.~\ref{sub:streamline_for_online_processing}. To keep the recovery on-grid, we first fix a uniform oversampled grid as the finest admissible candidate space. Specifically, in the delay domain, $\bar{\mathcal T}$ denotes the grid obtained by refining the DFT grid associated with full-band subcarrier observations by a super-resolution factor of $r_{\rm delay}^+$.

In dynamic channel tracking, the location parameters of each path usually evolve slowly. Hence, if the model over the previous window $\Omega_{t_0-1}$ based on the estimation space $\mathcal T_{t_0-1}$ yields a sparse estimate, then the estimated energy profile is also expected to remain concentrated around nearby grid points in the sparse regression for the next window $\Omega_{t_0}$. Motivated by this observation, we construct a local expansion space $\mathcal T'_{t_0}$ according to the previous delay-domain support as
\begin{equation}
    \mathcal T'_{t_0}:=\left\{\tau\in \bar{\mathcal T}: \{\tau-1,\tau,\tau+1\}\cap \bar E_{\mathcal T_{t_0-1}}^{k_{\max{}}}\neq\emptyset \right\},
\end{equation}
where $\bar E_{\mathcal T_{t_0-1}}^{k_{\max{}}}$ is obtained by embedding the delay-domain active support $E_{\mathcal T_{t_0-1}}^{k_{\max{}}}$, defined in Sec.~\ref{sub:f_t_subproblem} on $\mathcal T_{t_0-1}$, into $\bar{\mathcal T}$, with $k_{\max{}}$ being the maximum number of iterations. The space $\mathcal T'_{t_0}$ captures the local temporal continuity of the active subspace and allows the nonzero locations to be adjusted across slots at the resolution of $\bar {\mathcal T}$.

Nevertheless, local adjustment alone may miss newly appearing or previously uncovered paths caused by abrupt channel changes or inaccurate initialization. We therefore introduce an additional finite-budget subspace supplement by comparing the channel prediction with the latest observation. Specifically, if the estimate $X^{(t_0-1)}$ obtained from $\mathcal T_{t_0-1}$ has an evident grid-coverage deficiency, then the corresponding frequency-domain prediction $\tilde H^{(t_0|t_0-1)}_{\rm hop}$ should exhibit a noticeable mismatch with the observation $Y^{(t_0)}$ in the delay spectrum. After estimating the synchronization offsets, we compute the delay spectrum $R^{(t_0|t_0-1)}$ over $\bar{\mathcal T}$ from the residual $Y^{(t_0)}-\xi(\varepsilon^{(t_0)},\tau^{(t_0)})\tilde H^{(t_0|t_0-1)}_{\rm hop}$ using multi-timeslot multiple-signal-classification (MUSIC)~\cite{SchmidtMUSIC}, and choose the $r_{\rm delay}^{\rm p}$ grid points with the largest spectral values as $\mathcal T''_{t_0|t_0-1}$. Finally, to retain a coarse global anchor, we define $\mathcal T_{0}$ as a uniformly sampled coarse grid with undersampling factor $r_{\rm delay}^{-}$ relative to the same full-band DFT delay grid.
At the first slot, no historical information is available. We therefore set $\mathcal T_1:=\mathcal T_0\cup\mathcal T''_1$, where $\mathcal T''_1$ is determined solely from the delay spectrum of $Y^{(1)}$, using an enlarged value of $r_{\rm delay}^{\rm p}$.

By the structural symmetry of the model, the same estimation-space selection strategy is applied in the angular domain with the corresponding factors $r_{\rm ang}^+$, $r_{\rm ang}^-$, and $r_{\rm ang}^{\rm p}$. Super-resolution is performed along both the horizontal and vertical angular directions, and the resulting full grid is denoted by $\bar\varTheta$, in parallel with $\bar{\mathcal T}$. For the Doppler domain, the number of temporal observations is limited, and thus the refined Doppler space remains moderate in size. We therefore use a simple sampling-interval-dependent truncation rule. As implied by Proposition~\ref{prop:smooth_Doppler}, the Doppler spread induced by the channel variation over the observation window increases with the sounding interval $\delta t$. Hence, for a Doppler grid with fixed resolution, we specify a reference maximum grid size $N_{\mathrm d,\max{}}$ at the minimum sounding interval $\delta t_{\min}$, e.g., $\delta t_{\min}=5$ ms. For a sounding interval $\delta t$, we then set $\mathcal V_{t_0}:=\{1,\ldots,N_{\mathrm d,\max{}}\lfloor\delta t/\delta t_{\min}\rfloor\}$.

\subsection{Estimation of Residual Synchronization Offsets}\label{sub:non_ideal_MLE}

Although the estimation accuracy of the synchronization increments $\Delta\tau^{(t)}$ and $\Delta\varepsilon^{(t)}$ is not itself an evaluation target, large synchronization mismatches act as multiplicative perturbations and can substantially degrade multi-slot recovery. Nevertheless, they can be efficiently estimated via a two-parameter maximum-likelihood problem once the channel reconstructed at slot $t_0-1$ provides a reliable prediction $\tilde H_{\rm hop}^{(t_0|t_0-1)}$ over $\mathcal B_{t_0}$, the subband observed at slot $t_0$. With this prediction, the antenna dimension of the MIMO system supplies sufficient samples for synchronization, and the increments can be estimated from the least-squares fitting
\begin{equation}
    \min_{\Delta\tau,\Delta\varepsilon}
    \big\|Y^{(t_0)}
    -\xi(\varepsilon^{(t_0-1)}+\Delta\varepsilon,
          \tau^{(t_0-1)}+\Delta\tau)
      \tilde H_{\rm hop}^{(t_0|t_0-1)}
    \big\|_{\rm F}^2 .\label{eq:MLE_NonIdeal}
\end{equation}
Using the composition rule $\xi(\varepsilon+\varepsilon',\tau+\tau')=\xi(\varepsilon,\tau)\xi(\varepsilon',\tau')$ for any $\varepsilon,\varepsilon',\tau,\tau'$, and setting $U:=Y^{(t_0)}$ and $V:=\xi(\varepsilon^{(t_0-1)},\tau^{(t_0-1)})\tilde H_{\rm hop}^{(t_0|t_0-1)}$, the above problem reduces to $\min_{\Delta\tau,\Delta\varepsilon}\|U-\xi(\Delta\varepsilon,\Delta\tau)V\|_{\rm F}^2$. Since $\xi(\Delta\varepsilon,\Delta\tau)$ is unitary, this objective is equivalent to minimizing the cross term $-2\Re\{\langle \xi(\Delta\varepsilon,\Delta\tau)V,U\rangle_{\rm F}\}$. The reduced problem then admits a simple two-step solution.

Note that $\xi(\Delta\varepsilon,\Delta\tau)=e^{\mathrm i\Delta\varepsilon}\xi(0,\Delta\tau)$, where $\xi(0,\Delta\tau)$ is a diagonal delay-shift matrix with entries $\exp(-2\pi\mathrm i m\Delta\tau/N_{\rm f})$ for $m\in\mathcal B_{t_0}$. Then the objective becomes $-2\Re\{e^{-\mathrm i\Delta\varepsilon}g(\Delta\tau)\}$, where $g(\Delta\tau):=\langle \xi(0,\Delta\tau)V,U\rangle_{\rm F}$ has the Fourier form
\begin{equation}
    g(\Delta\tau)=\sum_{m\in\mathcal B_{t_0}}
      \Big(\sum_{r=1}^{N_{\rm r}}\overline{V_{m,r}}U_{m,r}\Big)
      e^{2\pi\mathrm i m\Delta\tau/N_{\rm f}} .
\end{equation}
Hence the two parameters can be optimized in a decoupled way. For any fixed $\Delta\tau$, the optimal $\Delta\varepsilon^*$ is given in closed form as $\Delta\varepsilon^*:=\Delta\varepsilon^*(\Delta\tau)=\angle g(\Delta\tau)$, where $\angle$ denotes the argument. Consequently, it suffices to determine $\Delta\tau^*$ by maximizing $|g(\Delta\tau)|$ through a one-dimensional search, after which the phase increment is given by $\Delta\varepsilon^*=\angle g(\Delta\tau^*)$. Since $g(\Delta\tau)$ is exactly the inverse Fourier transform of the per-subcarrier correlation sequence $\sum_r \overline{V_{r,m}}U_{r,m}$ after embedding it into the global frequency grid, this search can be implemented efficiently by an oversampled IFFT with suitable resolution. The estimated increments are then accumulated and used to form $\mathcal E_{t_0}$ in Sec.~\ref{sub:streamline_for_online_processing}, before the estimation-space determination $\mathcal J_{t_0}$ and the DDA recovery $\mathcal R_{t_0}$.

% In implementation, the first slot is used as the synchronization reference, i.e., $\tau^{(0)}=\varepsilon^{(0)}:=0$, since only relative alignment is required in tracking. During the warm-up stage with $t_0<T$, although the predicted coefficients in $\tilde H_{\rm hop}^{(t_0|t_0-1)}$ may still be inaccurate, their delay-domain energy profile can still provide a coarse correlation peak in~\eqref{eq:MLE_NonIdeal}. Thus the update helps reduce an otherwise uncalibrated global mismatch to a localized synchronization error, which in practice behaves more like a mild perturbation to the subsequent DDA recovery and can be further refined along the tracking process.

In implementation, the first slot is used as the synchronization reference, i.e., $\tau^{(0)}=\varepsilon^{(0)}:=0$, because only relative alignment is required for tracking. During the warm-up stage $t_0<T$, the predicted coefficients in $\tilde H_{\rm hop}^{(t_0|t_0-1)}$ may remain inaccurate, but their delay-domain energy profile can still provide a coarse correlation peak in~\eqref{eq:MLE_NonIdeal}. The update therefore reduces an otherwise uncalibrated global mismatch to a localized synchronization error, which behaves as a milder perturbation to the subsequent DDA recovery and can be progressively refined during tracking.

\section{Numerical Experiments}\label{sec:numerical_experiments}

\subsection{Comparison between Different Architectures}

We evaluate the proposed ST-DDA framework using geometry-based channels generated by QuaDRiGa 2.8.1~\cite{Quadriga} under the 3GPP TR~38.901 UMa-NLOS scenario in MATLAB R2026a. The carrier frequency is $3.5$ GHz, and the effective frequency-domain channel contains $408$ uniformly spaced frequency samples with a spacing of $240$ kHz. Each trajectory is generated with randomized UE geometry and mobility: the UE--BS distance is uniformly drawn from $[80,200]$ m within a $\pm60^\circ$ sector, the UE height is uniformly drawn from $[1.2,1.5]$ m, the direction of motion is uniformly drawn from $[0,2\pi)$, and the UE speed is uniformly drawn from $[2,6]$ km/h. The BS is fixed at a height of $20$ m and employs a 3GPP NR multi-panel antenna array with $N_x=4$, $N_y=8$, and $N_p=2$, while the UE employs an omnidirectional antenna. Different sounding intervals $\delta t$ are obtained by sampling each generated channel trajectory at the corresponding rates.

For the main evaluation, we separately sweep the SNR, the temporal sampling interval $\delta t$, and the frequency-hopping period $n_{\rm hop}$. The SNR sweep is produced by varying the additive-noise level, whereas the temporal-sampling sweep is produced by subsampling the snapshots along each channel trajectory. For the hopping-period sweep, we consider $n_{\rm hop}\in\{4,6,8,12,17\}$, partition the $N_{\rm f}=408$ frequency samples into $n_{\rm hop}$ subbands, and determine the sequence of active subbands $\mathcal B_{t_0}$ using the MCC hopping pattern~\cite{MCC}. Unless otherwise specified, the default parameters are ${\rm SNR}=20$ dB, $n_{\rm hop}=12$, and $\delta t=20$ ms, with only the parameter under investigation varied in each sweep.

For the proposed ST-DDA methods, we implement two variants to demonstrate both the basic $\ell_1$ architecture (ST-DDA-L1) and its reweighted variant (ST-DDA-RW).

\textbf{ST-DDA-L1}: The basic framework fixes all $\lambda_\zeta=\lambda_{t_0}$ during the iterations at the $t_0$-th sounding slot. For the estimation space, we set the parameters in Sec.~\ref{sub:estimation_space} as $r_{\rm delay}^+=3$, $r_{\rm ang}^+=1$, $r_{\rm delay}^-=r_{\rm ang}^-=4$, $r_{\rm delay}^{\rm p}=100$, and $r_{\rm ang}^{\rm p}=10$. For the adaptive $\lambda_{t_0}$, we set the budget support size as $150$, and set $\sigma_t:=\sqrt{2\tilde\sigma^2_{\Omega_{t_0}}M_{\rm f}N_{\rm r}|\Omega_{t_0}|\log(|\bar{\mathcal D}_{t_0}|)/|\bar{\mathcal D}_{t_0}|}$, where $\tilde\sigma^2_{\Omega_{t_0}}$ is the averaged noise power over $\Omega_{t_0}$ and $\bar{\mathcal D}_{t_0}:=\bar{\mathcal T}\times{\bar\varTheta}\times\mathcal V_{t_0}$. The maximum iteration number is set to $20$ for each slot. For a hopping period $n_{\rm hop}$, we set the window size as $T:=\max\{\lfloor 2n_{\rm hop}/3\rfloor,4\}$, the splitting factor as $c_T=3$ for $\Omega_{t_0}^+$ in Sec.~\ref{sub:online_processing}, and $r_T=\max\{\lfloor T/2-c_T\rfloor,1\}$ for $\Omega_{t_0}^-$ in the random dropout of the $\mathcal F_t$-update. We set the averaging factor $0.5$ for $H,\Pi$ and $0.9$ for $X$ in \eqref{line:AVE}.

\textbf{ST-DDA-RW} keeps the basic settings above and further introduces the reweighting step. For the reweighting neighborhood, we set $r_{\rm ang}^c=r_{\rm Dop}^c=2$ in~\eqref{eq:neighbor_theta}, \eqref{eq:neighbor_Dop}. For the randomized kernel $\widetilde{\mathcal N}_\zeta$, we randomly select $5$ points from the same-polarization angular neighborhood and $2$ points from the cross-polarization counterpart, always including the central point $\zeta$ and its polarized counterpart. In the Doppler domain, we randomly select half of $\mathcal N_{\zeta}^{\mathcal V}$ while keeping the two sides balanced in number. For the adaptive $\lambda_{t_0}$, we set the budget support size as $100$, and reduce the per-slot iteration budget to $10$. For stabilization, all regularization weights are fixed at $\lambda_{t_0}$ during the first $20$ sounding slots and during the first three iterations of each subsequent slot, so that the corresponding updates reduce to the basic $\ell_1$-based ST-DDA method.

We compare with the following baselines to examine the performance of different latest-slot recovery architectures.

\textbf{Freq-Extra}: Single-slot frequency extrapolation, which does not use historical information and estimates the delay--angular parameters solely from the current subband observation. Specifically, we use a quasi-Newton-based orthogonal matching pursuit (QNOMP)-type super-resolution method~\cite{QNOMP}, adapted to the $4\times 8$ dual-polarized UPA by assigning each physical path a common delay, two angular parameters, and two polarization-dependent complex gains. Local refinement is performed over the delay and two angular dimensions, and at most $50$ physical paths are retained.

\textbf{Time-Retain}: Subband-independent time-domain retaining, which updates only the observed subband at the latest sounding slot and directly retains the previous estimates on the unobserved subbands. For the observed subband, we apply DA-domain OMP~\cite{OMP} with a delay-domain super-resolution factor of three and retain at most $50$ nonzero elements.

\textbf{Dynamic CS}: We implement the PLAY-CS framework~\cite{PLAY-CS}, which incorporates historical information through an $\ell_2$ penalty over the predicted support while retaining an $\ell_1$ penalty elsewhere. At the $t_0$-th sounding slot, the Doppler-aware state transition first produces $X^{(t_0|t_0-1)}=\Psi_{t_0}(X^{(t_0-1)})$, after which $X^{(t_0)}$ is estimated by solving
\begin{align}
    \min_X&\frac1{2\sigma_{t_0}'}\|Y_{t_0}-A_{t_0}XF_{\rm r}^\top\|_{\rm F}^2+\mathcal W(X;E^{(t_0|t_0-1)},X^{(t_0|t_0-1)}),\notag\\
    &\mathcal W(X;E,C):=\lambda_1\|X|_{E^c}\|_1+\lambda_2\|[X-C]|_{E}\|_{\rm F}^2,\label{eq:PLAY_CS}
\end{align}
where $\|\cdot\|_1$ denotes the sum of elementwise moduli, $E^{(t_0|t_0-1)}$ contains the entries accounting for $95\%$ of the cumulative energy of $X^{(t_0|t_0-1)}$, and its remaining entries are truncated to zero. We define $\Psi_{t_0}(X^{(t_0-1)}):=e^{\mathrm{i}\boldsymbol{\nu}^{(t_0-1)}}\odot X^{(t_0-1)}$ and update the Doppler phase increment by $\boldsymbol \nu^{(t_0)}:=\alpha\angle(X^{(t_0)}\odot \overline{X^{(t_0-1)}})+(1-\alpha)\boldsymbol \nu^{(t_0-1)}$ with $\alpha=0.2$, where the overline denotes conjugation and $\odot$ denotes elementwise multiplication. The solution of~\eqref{eq:PLAY_CS} has been discussed in~\cite{ASM} via ASM. For fairness, we restrict the linear equation induced by the data fidelity to the same product-form subspace as in Sec.~\ref{sub:f_t_subproblem}, so that the same scheme~\eqref{eq:syl_scheme} and low-rank updates can be applied. We set $\lambda_1=\lambda_{t_0}$, $\lambda_2=0.2\lambda_{t_0}$, and $\sigma_{t_0}':=\sqrt{2\tilde\sigma^2_{t_0} M_{\rm f}N_{\mathrm r}\log(|\bar{\mathcal T}|\cdot|\bar \varTheta|)/(|\bar{\mathcal T}|\cdot|\bar \varTheta |)}$, where $\tilde \sigma^2_{t_0}$ is the noise power at the $t_0$-th slot and $\lambda_{t_0}$ follows the same budget-based adaptation as ST-DDA-RW while the reference support of $\lambda_{t_0}$ is set as $\{\tau\in\mathcal T_{t_0}: \exists\theta, {\rm s.t.}\ (\tau,\theta)\in E^{(t_0+1|t_0)} \}$. The estimation space uses the same parameters as ST-DDA-RW. To fully exploit PLAY-CS, we allow $50$--$300$ iterations and terminate after at least $50$ iterations once the proximal-gradient residual falls below $0.005$; see~\cite{ASM}.

For each trajectory, we track $100$ slots and use the median performance over the last $50$ steady-state slots for evaluation. The reported results are then obtained by taking the median over $200$ independent channel trajectories. We consider two metrics: quotient normalized mean-squared error (QNMSE) and spatial-domain correlation (SpCor). QNMSE accounts for phase and timing ambiguities by quotienting out the absolute phase and timing offsets, because these absolute offsets need not be identified when aligning the reconstructed frequency--spatial channel $\tilde H^{(t_0)}$ with its ground truth $\tilde H_*^{(t_0)}$. To this end, we define
\begin{equation*}
    {\rm QNMSE}(\tilde H^{(t_0)},\tilde H^{(t_0)}_*):=\min_{\varepsilon,\tau}\frac{\|\tilde H^{(t_0)}_*-\xi(\varepsilon,\tau)\tilde H^{(t_0)}\|_{\rm F}^2}{\|\tilde H^{(t_0)}_*\|_{\rm F}^2},
\end{equation*}
and report it in dB. The minimization is computed efficiently as in Sec.~\ref{sub:non_ideal_MLE}. Since Time-Retain estimates different subbands independently and does not provide cross-band synchronization, we optimize $(\varepsilon,\tau)$ for each subband separately. For SpCor, we evaluate the average normalized spatial correlation over all frequency tones, which is relevant to the beamforming gain~\cite{beamforming} and is invariant to the unitary factor $\xi(\varepsilon,\tau)$. Denoting by $\tilde h_n^{(t_0)}$ and $\tilde h_{*,n}^{(t_0)}$ the spatial channel vectors at the $n$-th frequency tone, we define
\begin{equation*}
    {\rm SpCor}(\tilde H^{(t_0)},\tilde H^{(t_0)}_*):=\frac1{N_{\rm f}}\sum_{n=1}^{N_{\rm f}}\left|\left(\frac{\tilde h_{*,n}^{(t_0)}}{\|\tilde h_{*,n}^{(t_0)}\|_2}\right)^{\!H}\frac{\tilde h_n^{(t_0)}}{\|\tilde h_n^{(t_0)}\|_2}\right|.
\end{equation*}

\begin{figure}[t]
\centering
\includegraphics[width=\columnwidth]{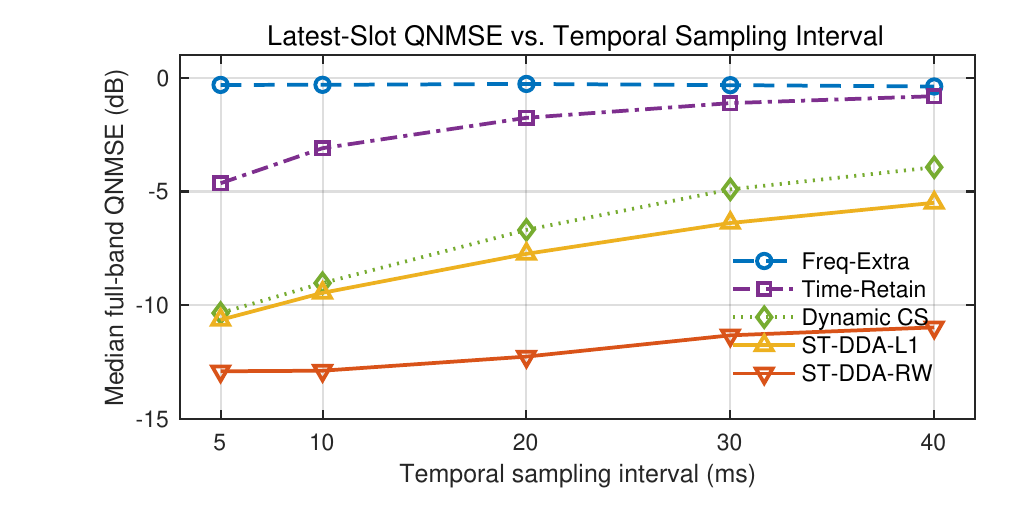}

\vspace{-1em}

\includegraphics[width=\columnwidth]{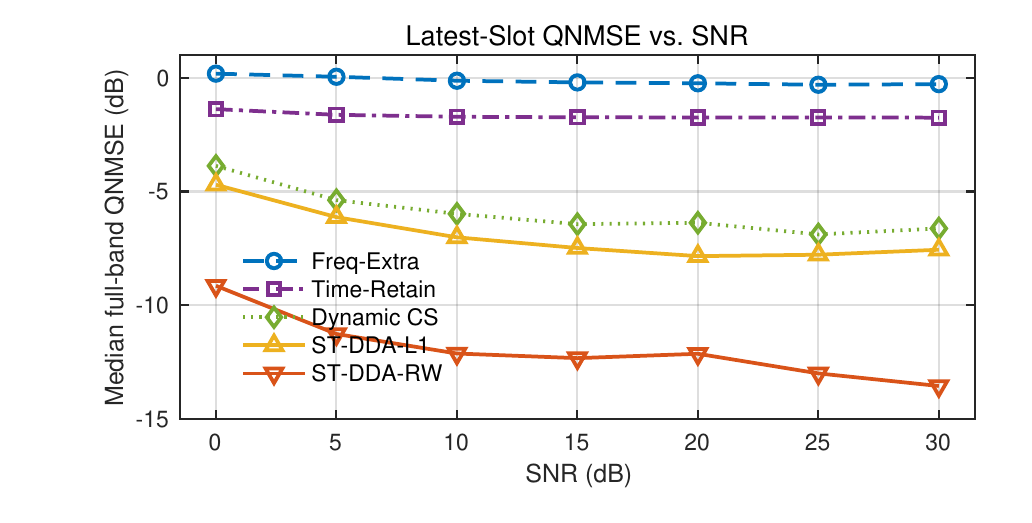}

\vspace{-1em}

\includegraphics[width=\columnwidth]{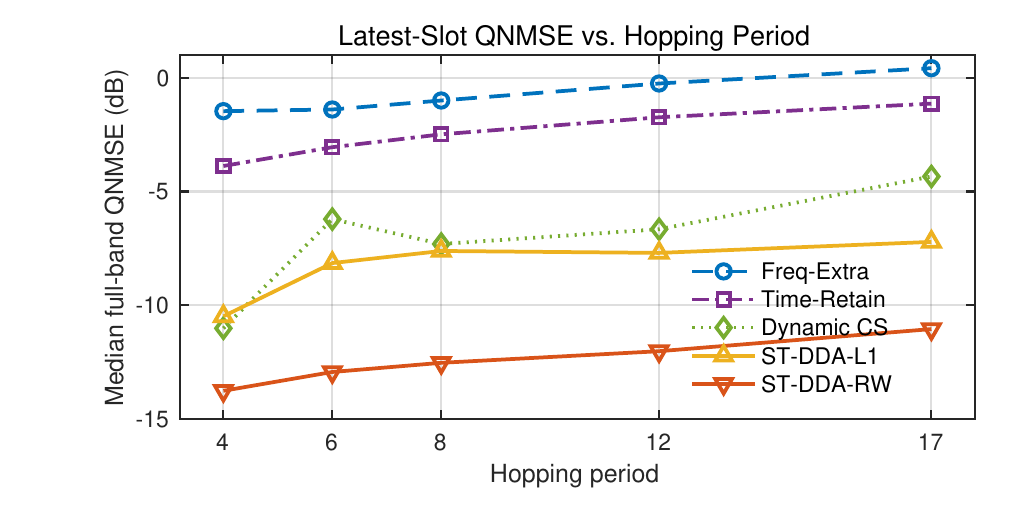}

\vspace{-1em}

\caption{Latest-slot full-band QNMSE comparison between different recovery architectures. 
The three panels sweep the temporal sampling interval $\delta t$, SNR, and hopping period $n_{\rm hop}$, respectively, while fixing the other parameters to their default values. 
Each point reports the median steady-state performance over $200$ QuaDRiGa UMa-NLOS trajectories.}
\label{fig:arch_qnmse}
\end{figure}

Figure~\ref{fig:arch_qnmse} first illustrates the difference between Dynamic CS and the multi-slot ST-DDA-L1 framework. Actually, when a user is allocated sufficient observation (small $n_{\rm hop}$) and the effective channel variation is slow (small $\delta t$), both frameworks perform alike and Dynamic CS can be preferable for its efficiency (see Table~\ref{tab:main_runtime}). As a single-slot filtering method, however, Dynamic CS applies only one Doppler-phase adjustment through $\Psi_{t_0}(\cdot)$ to each DA-domain coefficient and therefore cannot flexibly incorporate a distributed Doppler-domain prior. As $\delta t$ and $n_{\rm hop}$ increase, the per-slot observations become more limited and the channel varies more substantially, under which the multi-slot architectures exhibit greater robustness. The same trend is observed for SpCor, which is directly relevant to beamforming, and for the one-slot-ahead prediction results in Fig.~\ref{fig:arch_spcor_pred}, which affect the estimation of the residual synchronization offsets. We classify SpCor values above $0.8$ and $0.9$ as acceptable and successful, respectively, and report the corresponding trajectory-level rates in Table~\ref{tab:main_spcor_rate}. ST-DDA-RW maintains a substantially more stable success rate across the considered temporal sampling intervals than all baselines.

Although we have applied the ASM framework, which will be discussed later in Sec.~\ref{sub:ablations_ASM}, to alleviate the computational burden of the multi-slot model, ST-DDA-L1 still requires a larger budget intrinsic to the $\ell_1$-regularization. Different from Dynamic CS and ST-DDA-RW, the support budget of $\lambda_{t_0}$ is increased for accuracy since multi-slot joint processing naturally leads to support inflation and a simple $\ell_1$-model can not effectively rule out them. After introducing the position-embedded reweighting, ST-DDA-RW narrows the support range and, as implied in~\cite{ASM}, a stronger structured-sparsity-inducing denoiser empirically help converge faster allowing earlier stopping at each slot (see Table~\ref{tab:main_runtime}).

\begin{table}[h]
\caption{Per-slot steady-state runtime of different architectures.}
\label{tab:main_runtime}

\vspace{-1.0ex}

\centering
\scriptsize
\setlength{\tabcolsep}{2.0pt}
\renewcommand{\arraystretch}{1.05}
\resizebox{\columnwidth}{!}{%
\begin{tabular}{@{}lccccc@{}}
\toprule
Setting & Freq-Extra & Time-Retain & Dynamic CS & ST-DDA-L1 & ST-DDA-RW \\
\midrule
Default        & 0.1957 & 0.01765 & 0.05917 & 0.1665 & 0.1063  \\
Short interval & 0.1958 & 0.01774 & 0.09031 & 0.1124 & 0.06481 \\
Long interval  & 0.1956 & 0.01771 & 0.04851 & 0.2505 & 0.1589  \\
Small hopping  & 0.2716 & 0.04907 & 0.1856  & 0.1868 & 0.1129  \\
Large hopping  & 0.1719 & 0.01278 & 0.04573 & 0.1812 & 0.1118  \\
Low SNR         & 0.03338 & 0.01759 & 0.05811 & 0.1632 & 0.1163  \\
High SNR        & 0.1954 & 0.01777 & 0.1225  & 0.1673 & 0.1076  \\
\bottomrule
\end{tabular}%
}
\vspace{0.5ex}

\parbox{\columnwidth}{\scriptsize\raggedright
\emph{Note:} The default setting is
$(\delta t,{\rm SNR},n_{\rm hop})=(20{\rm ms},20{\rm dB},12)$.
The short/long-interval, small/large-hopping, and low/high-SNR cases set
$\delta t=5/40$ ms, $n_{\rm hop}=4/17$, and ${\rm SNR}=0/30$ dB, respectively,
with the other parameters fixed at their default values. Each entry is in
seconds and is computed as the median over $200$ trajectories of the average per-slot
runtime over the last 50 steady-state tracking slots.
}
\end{table}

\vspace{-3.5ex}

\begin{table}[h]
\caption{Trajectory-level successful/acceptable recovery rates in terms of SpCor under varying temporal sampling intervals.}
\label{tab:main_spcor_rate}

\vspace{-1.0ex}

\centering
\scriptsize
\setlength{\tabcolsep}{2.0pt}
\renewcommand{\arraystretch}{1.05}
\resizebox{\columnwidth}{!}{%
\begin{tabular}{@{}cccccc@{}}
\toprule
$\delta t$ (ms) & Freq-Extra & Time-Retain & Dynamic CS & ST-DDA-L1 & ST-DDA-RW \\
\midrule
5  & 0.0/0.0 & 21.5/53.0 & 84.0/97.5 & 95.0/99.0 & \textbf{97.0/99.5} \\
10 & 0.0/0.0 & 10.5/34.5 & 70.0/92.5 & 90.5/98.0 & \textbf{92.0/98.5} \\
20 & 0.0/0.0 &  6.0/19.0 & 43.0/70.5 & 72.0/94.5 & \textbf{87.5/96.0} \\
30 & 0.0/0.5 &  4.5/14.5 & 28.5/59.0 & 52.5/87.5 & \textbf{90.5/99.0} \\
40 & 0.0/0.5 &  2.0/9.0  & 14.5/49.0 & 37.0/83.0 & \textbf{93.5/99.0} \\
\bottomrule
\end{tabular}%
}
\vspace{0.5ex}

\parbox{\columnwidth}{\scriptsize\raggedright
\emph{Note:} Each entry reports the successful/acceptable rate in percent,
where a trajectory is classified as successful or acceptable if its median
SpCor over the last 50 slots exceeds $0.9$ or $0.8$, respectively.
The rates are computed over $200$ independent trajectories.
}
\end{table}

\begin{figure}[t]
\centering
\includegraphics[width=\columnwidth]{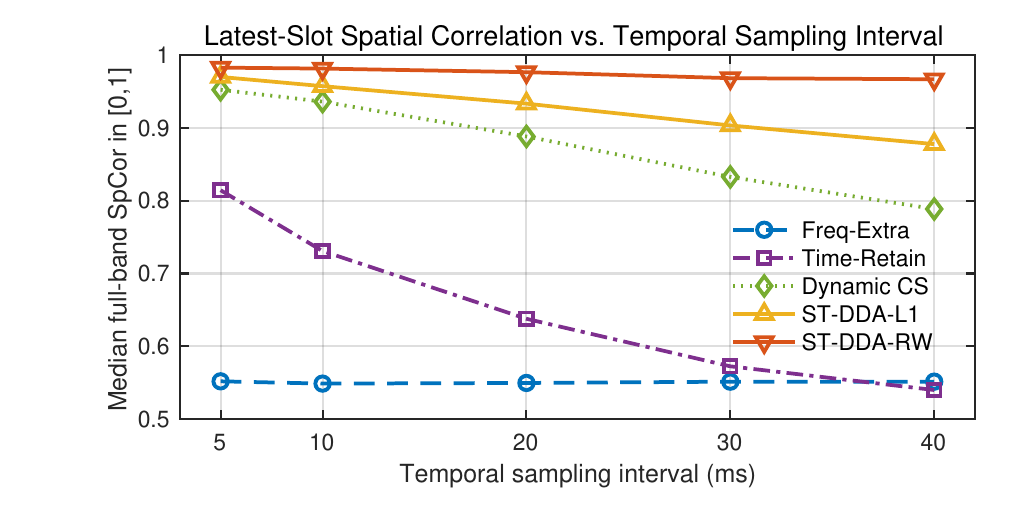}

\vspace{-1em}

\includegraphics[width=\columnwidth]{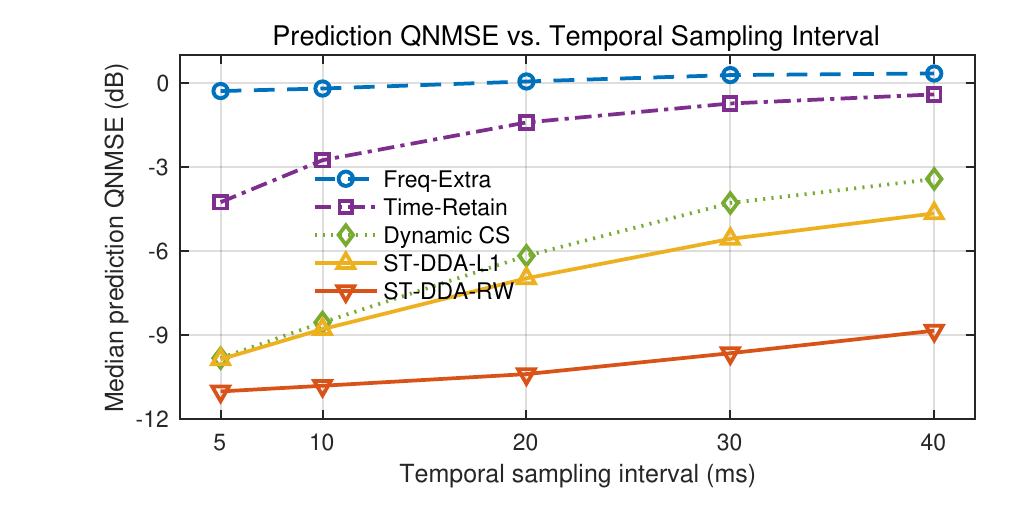}

\vspace{-1em}

\caption{Performance comparison under different temporal sampling intervals. 
[Top]: latest-slot full-band spatial correlation, where higher values indicate better reconstruction. 
[Bottom]: single-slot prediction QNMSE before assimilating the latest subband observation.
Each point reports the median steady-state performance over $200$ QuaDRiGa UMa-NLOS trajectories.}
\label{fig:arch_spcor_pred}

% \vspace{-1em}
\end{figure}

\subsection{Additional Experiments on Statistical Channels}

We further consider a randomized CDL-B statistical ensemble in which the UE speed is uniformly drawn from $[1,4]$ km/h and the RMS delay spread is uniformly drawn from $[50,200]$ ns. This ensemble complements the QuaDRiGa-based geometry simulations with a standardized statistical NLOS test and introduces rich delay-domain dispersion for DDA-sparse recovery. Under this setting, we additionally evaluate two variants of ST-DDA-RW. The first, denoted by \textbf{ST-DDA-Ang}, uses only the angular reweighting neighborhood, i.e., $\mathcal N_\zeta=\mathcal N_\zeta^{\varTheta}$, and omits the Doppler neighbors. The second, denoted by \textbf{ST-DDA-DopExt}, retains both the angular and Doppler neighborhoods but enlarges the Doppler reweighting radius to $r_{\rm Dop}^c=6$ to examine the influence of a wider Doppler kernel at large sampling intervals. All other parameters are kept identical to those used in the preceding experiments.

As shown in Fig.~\ref{fig:cdlb_ablation}, when the sampling interval is small, most of the reweighting gain arises from the angular-domain structure. As the sampling interval increases, channel variation produces a wider Doppler-domain representation, and methods that incorporate Doppler-domain group sparsity become more robust. The results of ST-DDA-DopExt further indicate that the recovery accuracy is not highly sensitive to the precise Doppler reweighting radius, although the best-performing radius may vary across operating conditions.

\begin{figure}[t]
\centering
\includegraphics[width=\columnwidth]{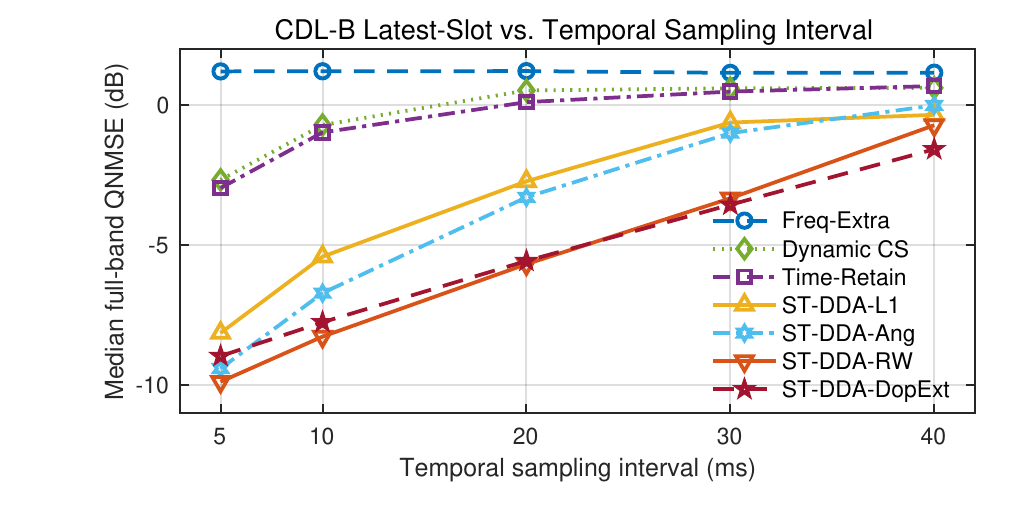}

\vspace{-1em}

\caption{Latest-slot full-band QNMSE comparison on CDL-B channels with respect to the temporal sampling interval $\delta t$ (SNR=$20$dB, $n_{\rm hop}=12$). Each point reports the median steady-state performance over 200 trajectories.}
\label{fig:cdlb_ablation}
\end{figure}

\subsection{Ablations on the Choice of Solvers}\label{sub:ablations_ASM}

ST-DDA-RW is a multi-slot reweighted architecture and is not tied to the ASM solver in Sec.~\ref{sec:ASM_Iteration}. The reweighting rule and the random dropout strategy can also be combined with other solvers. We therefore compare the following variants. (i) \textbf{PDHG}: the primal--dual hybrid gradient (PDHG) method, implemented using the Chambolle--Pock iteration~\cite{C-P}, solves~\eqref{eq:SingleWindowModel} without matrix inversions and is applicable to irregular hopping patterns. (ii) \textbf{ADMM}: the full-space counterpart of ASM, where the $\mathcal F_t$-subproblem and the fidelity step~\eqref{eq:sub_data_fidelity} in the $\mathcal G_\lambda$-subproblem are both solved over the full estimation space $\mathcal D_{t_0}$. Since the full-space operator is fixed at each slot, the inverse matrices in~\eqref{eq:syl_scheme} can be precomputed and reused. (iii) \textbf{ADMM-ASM}: an ablation of ST-DDA-RW that replaces only the $\mathcal F_t$-subproblem by its full-space version. (iv) \textbf{ASM-ISTA}: an ablation that removes the fidelity step~\eqref{eq:sub_data_fidelity} from the $\mathcal G_\lambda$-subproblem, so that this subproblem is solved by a simple ISTA iteration~\eqref{eq:sub_ISTA}.

We vary the delay-domain undersampling rate $r_{\rm delay}^-$ and define the undersampling factor (UF) as $|\bar{\mathcal T}|/|\mathcal T_0|$ to evaluate the accuracy under different degrees of estimation-space completeness (e.g., $r_{\rm delay}^-=4$ and $r_{\rm delay}^+=3$ lead to UF$=12$). Reducing this factor enlarges the base estimation grid and thus serves as a controlled proxy for the estimation-space inflation that may occur during channel tracking, allowing us to examine the corresponding runtime growth. At each slot, all methods are assigned a budget of $10$ iterations, except PDHG, which is assigned $20$. 

As shown in Fig.~\ref{fig:Ablation_Grid_Size}, despite the doubled budget, the slower convergence of PDHG prevents it from attaining the steady-state accuracy of the splitting-based variants, whose RLS-based data-fidelity steps more effectively approach the window-model minimizer. The same effect is observed within the $\mathcal G_\lambda$-subproblem: removing the fidelity step leads to the substantial performance degradation of ASM-ISTA. The remaining full-fidelity variants achieve accuracy competitive with the fully ASM-based ST-DDA-RW. Nevertheless, as shown in Table~\ref{tab:solver_undersampling_runtime}, restricting the fidelity steps to selected subspaces allows ST-DDA-RW to exhibit substantially milder runtime growth as the estimation space inflates.

\begin{figure}[t]
\centering
\includegraphics[width=\columnwidth]{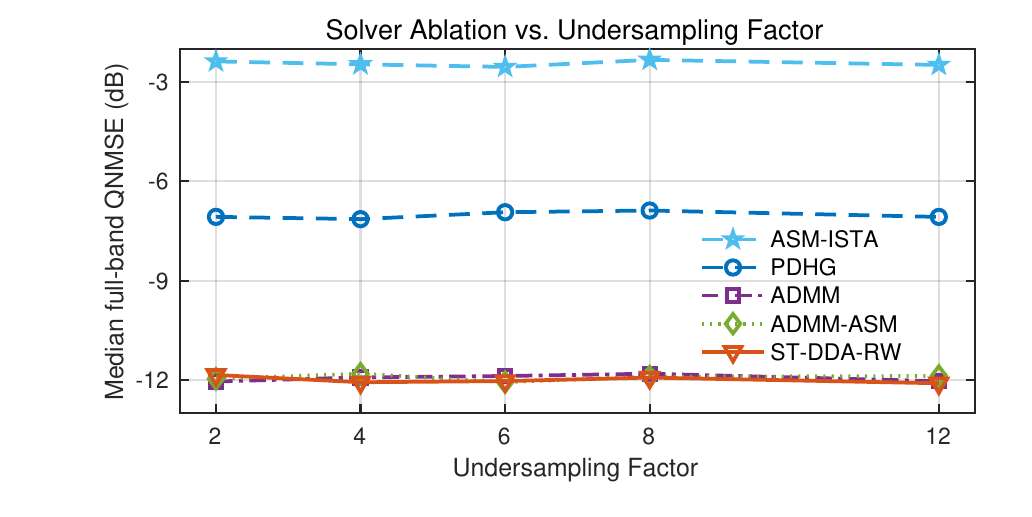}

\vspace{-.5em}

\includegraphics[width=\columnwidth]{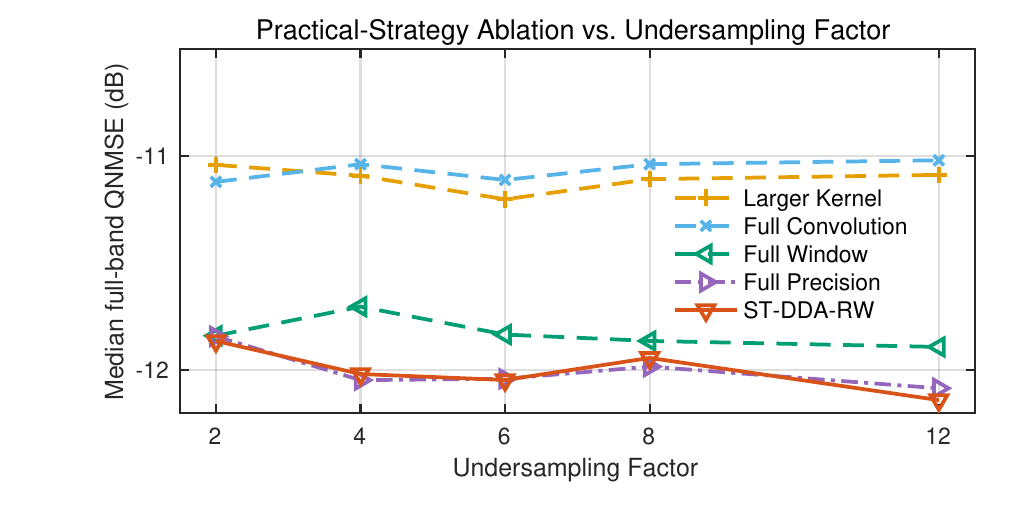}

\vspace{-1em}

\caption{Latest-slot full-band QNMSE on QuaDRiGa UMa-NLOS channels versus the delay-domain UF at SNR $=20$ dB, $n_{\rm hop}=12$, and $\delta t=20$ ms: [top] solver ablation; [bottom] practical-strategy ablation. Each point reports the median steady-state performance over $200$ trajectories.}
\label{fig:Ablation_Grid_Size}
\par\vspace{-1.5ex}
\end{figure}

\vspace{-1.5ex}

\begin{table}[h]
\caption{Per-slot steady-state runtime of different solver variants under varying delay undersampling factors (UF).}
\label{tab:solver_undersampling_runtime}

\vspace{-1.5ex}

\centering
\scriptsize
\setlength{\tabcolsep}{4.0pt}
\renewcommand{\arraystretch}{1.05}
\begin{tabular}{@{}cccccc@{}}
\toprule
UF & ASM-ISTA & PDHG & ADMM & ADMM-ASM & ST-DDA-RW \\
\midrule
2  & 0.2529 & 0.4059 & 0.4201 & 0.3224 & \textbf{0.1882} \\
4  & 0.1690 & 0.2811 & 0.2301 & 0.1799 & \textbf{0.1410} \\
6  & 0.1404 & 0.2416 & 0.1773 & 0.1433 & \textbf{0.1233} \\
8  & 0.1275 & 0.2184 & 0.1547 & 0.1276 & \textbf{0.1144} \\
12 & 0.1155 & 0.1922 & 0.1325 & 0.1120 & \textbf{0.1060} \\
\bottomrule
\end{tabular}
\par\vspace{.8ex}
\parbox{\columnwidth}{\scriptsize\raggedright
\emph{Note:} Each entry is in seconds and is computed as the median over $200$ trajectories of the average per-slot runtime over the last $50$ steady-state tracking slots. The best runtime under each undersampling factor is highlighted in bold.
}
\end{table}
\par\vspace{-1.5ex}

\subsection{Ablations on the Practical Strategies for Efficiency}

We next ablate the practical strategies introduced for computational efficiency. Since all compared methods are variants of ST-DDA-RW, each is named after the corresponding modification:
(i) \textbf{Full Window}, which disables the random dropout in Sec.~\ref{sub:online_processing} and uses the complete window $\Omega_{t_0}$ at every iteration;
(ii) \textbf{Full Convolution}, which uses the complete reweighting neighborhood, i.e., $\widetilde{\mathcal N}_\zeta=\mathcal N_\zeta$ for every $\zeta$;
(iii) \textbf{Full Precision}, which performs the entire reweighting procedure in full precision; and
(iv) \textbf{Larger Kernel}, which enlarges both the Doppler and the angular convolutional kernel to $r_{\rm Dop}^c=8$ and $r_{\rm ang}^c=3$ to evaluate the sensitivity to the kernel size.

The accuracy versus the delay-domain UF is reported in Fig.~\ref{fig:Ablation_Grid_Size}, with the corresponding runtime given in Table~\ref{tab:practical_strategy_runtime}. The proposed practical strategies reduce the computational cost without appreciable accuracy loss. Moreover, the accuracy is not highly sensitive to the kernel size, as even an overly large kernel causes only mild performance degradation; nevertheless, a larger kernel is not recommended as the default because of its additional computational cost.

\vspace{-1.5ex}
\begin{table}[h]
\caption{Per-slot steady-state runtime of practical-strategy
variants under varying delay-domain undersampling factors
(UF).}
\label{tab:practical_strategy_runtime}
\vspace{-1.5ex}
\centering
\scriptsize
\setlength{\tabcolsep}{3.0pt}
\renewcommand{\arraystretch}{1.05}
\begin{tabular}{@{}cccccc@{}}
\toprule
UF
& \shortstack{Larger Kernel}
& \shortstack{Full Convolution}
& \shortstack{Full Window}
& \shortstack{Full Precision}
& ST-DDA-RW \\
\midrule
2  & 0.2066 & 0.2028 & 0.2199 & 0.1914 & \textbf{0.1865} \\
4  & 0.1583 & 0.1550 & 0.1622 & 0.1441 & \textbf{0.1381} \\
6  & 0.1358 & 0.1387 & 0.1436 & 0.1275 & \textbf{0.1215} \\
8  & 0.1263 & 0.1295 & 0.1340 & 0.1180 & \textbf{0.1121} \\
12 & 0.1175 & 0.1200 & 0.1248 & 0.1097 & \textbf{0.1042} \\
\bottomrule
\end{tabular}
\par\vspace{.8ex}
\parbox{\columnwidth}{\scriptsize\raggedright
\emph{Note:} Each entry is in seconds and is computed as the
median over $200$ trajectories of the average per-slot runtime
over the last $50$ steady-state tracking slots. The best runtime
under each undersampling factor is highlighted in bold.
}
\end{table}

% \vspace{-1.5ex}

% \begin{table}[h]
% \caption{Per-slot steady-state runtime of practical-strategy variants under varying temporal sampling intervals.}
% \label{tab:practical_strategy_runtime}

% \vspace{-1.5ex}

% \centering
% \scriptsize
% \setlength{\tabcolsep}{3.2pt}
% \renewcommand{\arraystretch}{1.05}
% \begin{tabular}{@{}cccccc@{}}
% \toprule
% \shortstack{$\delta t$ (ms)}
% & \shortstack{Larger-Kernel}
% & \shortstack{Full-Convolution}
% & \shortstack{Full-Window}
% & \shortstack{Full-Precision}
% & ST-DDA-RW \\
% \midrule
% 5  & 0.0688 & 0.0718 & 0.0838 & 0.0628 & \textbf{0.0619} \\
% 10 & 0.0818 & 0.0854 & 0.0948 & 0.0766 & \textbf{0.0733} \\
% 20 & 0.1162 & 0.1180 & 0.1232 & 0.1089 & \textbf{0.1032} \\
% 30 & 0.1597 & 0.1527 & 0.1537 & 0.1419 & \textbf{0.1340} \\
% 40 & 0.1865 & 0.1766 & 0.1733 & 0.1637 & \textbf{0.1540} \\
% \bottomrule
% \end{tabular}
% \par\vspace{.8ex}
% \parbox{\columnwidth}{\scriptsize\raggedright
% \emph{Note:} Each entry is in seconds and is computed as the median over $200$ trajectories of the average per-slot runtime over the last $50$ steady-state tracking slots. The best runtime at each temporal sampling interval is highlighted in bold.
% }
% \end{table}

\section{Conclusion}

Dynamic CS provides a practical route for exploiting historical channel information. Although joint multi-slot processing with DDA-domain priors has shown strong potential in various settings, its application to frequency-hopping systems remains challenging. This work showed that slowly varying channels can retain Doppler-domain energy concentration over a finite window, thereby supporting a multi-slot DDA-sparse recovery model. To alleviate the resulting computational burden, our local stability analysis motivated a cross-slot online architecture in which each slot retains a well-defined sparse objective while inheriting the preceding-window estimate. This modular formulation allows efficient sparse solvers and structured priors to be incorporated separately. By combining ASM with convolutional reweighting, ST-DDA-RW behaves more robust than dynamic CS at a comparable per-slot computational scale, particularly for longer sounding intervals and larger hopping periods. These results establish ST-DDA as a competitive architecture for frequency-hopping OFDM systems and provide a basis for further algorithmic improvements.

\appendices

\section{Proof of Proposition~\ref{prop:smooth_Doppler}}\label{app:prop_smooth_doppler}

Let $\psi_\nu(t):=\vp(t)-2\pi \nu t$. Since $\vp'(t)/(2\pi)\in[\nu_{\min},\nu_{\max}]$ and $\nu\in\mathcal I_{\Delta\nu}^c$, we have $|\vp'(t)/(2\pi)-\nu|>\Delta\nu$ for all $t\in[0,T_{\rm w}]$, and hence $|\psi_\nu'(t)|\ge 2\pi\Delta\nu$. By $y(t)=x(t)e^{\mathrm i\vp(t)}$, we write $\tilde x(\nu)=T_{\rm w}^{-1}\int_0^{T_{\rm w}}x(t)e^{\mathrm i\psi_\nu(t)}\,{\rm d}t$. Using $e^{\mathrm i\psi_\nu(t)}=(\mathrm i\psi_\nu'(t))^{-1}{\rm d}(e^{\mathrm i\psi_\nu(t)})/{\rm d}t$ and integrating by parts, we obtain
\begin{equation*}
    \tilde x(\nu)=\frac{1}{T_{\rm w}}\left[\frac{x(t)e^{\mathrm i\psi_\nu(t)}}{\mathrm i\psi_\nu'(t)}\right]_0^{T_{\rm w}}
    -\frac{1}{T_{\rm w}}\int_0^{T_{\rm w}}e^{\mathrm i\psi_\nu(t)}
    \frac{{\rm d}}{{\rm d}t}\left[\frac{x(t)}{\mathrm i\psi_\nu'(t)}\right]{\rm d}t .
\end{equation*}
The boundary term is bounded by $c_0/(T_{\rm w}\pi\Delta\nu)$. Moreover, $|\frac{{\rm d}}{{\rm d}t}[x(t)/(\mathrm i\psi_\nu'(t))]|
\leq |x'(t)|/|\psi_\nu'(t)|+|x(t)||\psi_\nu''(t)|/|\psi_\nu'(t)|^2
\leq c_0 r_1/(2\pi\Delta\nu)+c_0 r_2/(2\pi\Delta\nu)^2$, where $\psi_\nu''(t)=\vp''(t)$. Combining the two bounds yields~\eqref{eq:prop1}.

\section{Proof of Proposition~\ref{prop:parent_child_error}}\label{app:prop_parent_child_error}

Define $q_-:=\nabla\ell_{T+1}(\hat x_0)$ and $q_+:=\nabla\ell_1(\hat x_0)$, and let $h_+:=\hat x_+-\hat x_0$ and $h_-:=\hat x_--\hat x_0$. Since $\operatorname{supp}(\hat x_0)\subseteq \mathcal S$ and $\|z_0|_{\mathcal S^c}\|_\infty\le 1-\gamma$, the $\ell_1$-Bregman divergence $D^{z_0}_{1}(\hat x_0+h,\hat x_0):=\|\hat x_0+h\|_1-\|\hat x_0\|_1-\Re\{z_0^Hh\}$ satisfies
\begin{equation}
    D^{z_0}_{1}(\hat x_0+h,\hat x_0)
    \geq \gamma\|h_{S^c}\|_1,\qquad (\forall h).
    \label{eq:bregman_margin_short}
\end{equation}
Indeed, $D_1^{z_0}(\hat x_0+h,\hat x_0)$ is separable across coordinates. Its contribution over $\mathcal S$ is nonnegative, whereas over $\mathcal S^c$ we have $\hat x_0|_{\mathcal S^c}=0$ and $|h_j|-\Re\{\overline{z_{0,j}}h_j\}\geq\gamma|h_j|$ for every $j\in\mathcal S^c$.

The optimality of $\hat x_-$ gives $\sum_{t\in\Omega_-}\ell_{t}(\hat x_-)+\lambda\|\hat x_-\|_1\leq\sum_{t\in\Omega_-}\ell_{t}(\hat x_0)+\lambda\|\hat x_0\|_1$, which equivalently leads to
\begin{equation}
    \frac12\sum_{t=1}^{T}|(F_{\mathrm d})_t h_-|^2
    +\lambda D^{z_0}_{1}(\hat x_0+h_-,\hat x_0)
    \leq
    \Re\{q_-^Hh_-\}.
    \label{eq:child_basic_minus}
\end{equation}

Combining \eqref{eq:bregman_margin_short} and \eqref{eq:child_basic_minus}, we know $\lambda\gamma\|h_-|_{\mathcal S^c}\|_1\leq\delta_-\|h_-\|_1$, where we use the fact that $\sum_{t=1}^{T}\nabla\ell_t(\hat x_0)+\lambda z_0=-\nabla\ell_{T+1}(\hat x_0)=-q_-$, by the parent KKT identity, and $\Re\{q_-^Hh_-\}\leq \delta_-\|h_-\|_1$. Since $\delta_-\leq\lambda\gamma/2$, we have $\|h_-|_{\mathcal S^c}\|_1\leq \|h_-|_\mathcal S\|_1$. Then by the restricted curvature on the shared slots and again using \eqref{eq:child_basic_minus}, we get
\begin{equation}
    \frac{\kappa_c}{2}\|h_-\|_2^2
    \leq
    \delta_-\|h_-\|_1
    \leq
    \delta_-2\sqrt{|\mathcal S|}\|h_-\|_2 .
\end{equation}
Hence $\|h_-\|_2 \leq 4\delta_-\kappa_c^{-1}\sqrt{|\mathcal S|}$. Applying the same argument to $\hat x_+$ gives $\|h_+\|_2\leq4\delta_+\kappa_c^{-1}\sqrt{|\mathcal S|}$. The desired result then follows from $\|\hat x_+-\hat x_-\|_2\leq\|h_+\|_2+\|h_-\|_2$.

\section*{Acknowledgment}

The authors acknowledge the support from National Key R\&D Program of China under grant 2021YFA1003301, and National Science Foundation of China under grant 12288101.

\bibliographystyle{IEEEtran}
\bibliography{IEEEabrv, STDDA}

\begin{thebibliography}{10}
\providecommand{\url}[1]{#1}
\csname url@samestyle\endcsname
\providecommand{\newblock}{\relax}
\providecommand{\bibinfo}[2]{#2}
\providecommand{\BIBentrySTDinterwordspacing}{\spaceskip=0pt\relax}
\providecommand{\BIBentryALTinterwordstretchfactor}{4}
\providecommand{\BIBentryALTinterwordspacing}{\spaceskip=\fontdimen2\font plus
\BIBentryALTinterwordstretchfactor\fontdimen3\font minus
  \fontdimen4\font\relax}
\providecommand{\BIBforeignlanguage}[2]{{%
\expandafter\ifx\csname l@#1\endcsname\relax
\typeout{** WARNING: IEEEtran.bst: No hyphenation pattern has been}%
\typeout{** loaded for the language `#1'. Using the pattern for}%
\typeout{** the default language instead.}%
\else
\language=\csname l@#1\endcsname
\fi
#2}}
\providecommand{\BIBdecl}{\relax}
\BIBdecl

\bibitem{CS}
D.~Donoho, ``Compressed sensing,'' \emph{IEEE Transactions on Information
  Theory}, vol.~52, no.~4, pp. 1289--1306, Apr. 2006.

\bibitem{CSTao}
E.~J. Cand{\`e}s, J.~K. Romberg, and T.~Tao, ``Stable signal recovery from
  incomplete and inaccurate measurements,'' \emph{Communications on Pure and
  Applied Mathematics}, vol.~59, no.~8, pp. 1207--1223, Mar. 2006.

\bibitem{3gpp_ts38211_v1850}
{ETSI 3rd Generation Partnership Project (3GPP)}, ``{5G; NR; Physical channels
  and modulation (3GPP TS 38.211 version 18.5.0 Release 18)},'' ETSI, Technical
  Specification ETSI TS 138 211 V18.5.0, Jan. 2025, release 18.

\bibitem{MCC}
X.~Zhu, Y.~Zeng, and T.~Li, ``Coverage-distance and collinearity-minimizing
  pilots for channel estimation in tdd systems,'' \emph{IEEE Wireless
  Communications Letters}, vol.~15, pp. 3741--3745, 2026.

\bibitem{candes}
E.~J. Candès and C.~Fernandez-Granda, ``Towards a mathematical theory of
  super-resolution,'' \emph{Communications on Pure and Applied Mathematics},
  vol.~67, no.~6, pp. 906--956, 2014.

\bibitem{4518398}
M.~Sharp and A.~Scaglione, ``Application of sparse signal recovery to
  pilot-assisted channel estimation,'' in \emph{2008 IEEE International
  Conference on Acoustics, Speech and Signal Processing}, 2008, pp. 3469--3472.

\bibitem{berger2009sparse}
C.~R. Berger, S.~Zhou, W.~Chen, and P.~Willett, ``Sparse channel estimation for
  {OFDM}: Over-complete dictionaries and super-resolution,'' in \emph{2009 IEEE
  10th Workshop on Signal Processing Advances in Wireless
  Communications}.\hskip 1em plus 0.5em minus 0.4em\relax IEEE, 2009, pp.
  196--200.

\bibitem{5454399}
W.~U. Bajwa, J.~Haupt, A.~M. Sayeed, and R.~Nowak, ``Compressed channel
  sensing: A new approach to estimating sparse multipath channels,''
  \emph{Proceedings of the IEEE}, vol.~98, no.~6, pp. 1058--1076, 2010.

\bibitem{KalmanCS}
N.~Vaswani, ``Kalman filtered compressed sensing,'' in \emph{2008 15th IEEE
  International Conference on Image Processing}, 2008, pp. 893--896.

\bibitem{PLAY-CS}
X.~Liu and Y.~Xia, ``A unified algorithmic framework for dynamic compressive
  sensing,'' \emph{Signal Processing}, vol. 232, p. 109926, 2025.

\bibitem{RM-BPDN}
W.~Lu and N.~Vaswani, ``Regularized modified {BPDN} for noisy sparse
  reconstruction with partial erroneous support and signal value knowledge,''
  \emph{IEEE Transactions on Signal Processing}, vol.~60, no.~1, pp. 182--196,
  2012.

\bibitem{6638908}
A.~S. Charles and C.~J. Rozell, ``Dynamic filtering of sparse signals using
  reweighted $\ell$1,'' in \emph{2013 IEEE International Conference on
  Acoustics, Speech and Signal Processing}, 2013, pp. 6451--6455.

\bibitem{ModifiedCS}
N.~Vaswani and W.~Lu, ``Modified-{CS}: Modifying compressive sensing for
  problems with partially known support,'' in \emph{2009 IEEE International
  Symposium on Information Theory}, 2009, pp. 488--492.

\bibitem{6557543}
J.~Ziniel and P.~Schniter, ``Dynamic compressive sensing of time-varying
  signals via approximate message passing,'' \emph{IEEE Transactions on Signal
  Processing}, vol.~61, no.~21, pp. 5270--5284, 2013.

\bibitem{10311526}
Y.~Wan, G.~Liu, A.~Liu, and M.-J. Zhao, ``Robust multi-user channel tracking
  scheme for {5G} new radio,'' \emph{IEEE Transactions on Wireless
  Communications}, vol.~23, no.~6, pp. 5878--5894, 2024.

\bibitem{2Stage}
Y.~Wan and A.~Liu, ``A two-stage {2D} channel extrapolation scheme for tdd {5G}
  {NR} systems,'' \emph{IEEE Transactions on Wireless Communications}, vol.~23,
  no.~8, pp. 8497--8511, 2024.

\bibitem{P7_OTFS}
R.~Hadani, S.~Rakib, M.~Tsatsanis, A.~Monk, A.~J. Goldsmith, A.~F. Molisch, and
  R.~Calderbank, ``Orthogonal time frequency space modulation,'' in \emph{2017
  IEEE wireless communications and networking conference (WCNC)}.\hskip 1em
  plus 0.5em minus 0.4em\relax IEEE, 2017, pp. 1--6.

\bibitem{P8_raviteja2019embedded}
P.~Raviteja, K.~T. Phan, and Y.~Hong, ``Embedded pilot-aided channel estimation
  for {OTFS} in delay--{D}oppler channels,'' \emph{IEEE transactions on
  vehicular technology}, vol.~68, no.~5, pp. 4906--4917, 2019.

\bibitem{TensorBased}
H.~Hou, Y.~Wang, Y.~Zhu, X.~Yi, W.~Wang, D.~T.~M. Slock, and S.~Jin, ``A
  tensor-structured approach to dynamic channel prediction for massive {MIMO}
  systems with temporal non-stationarity,'' \emph{IEEE Transactions on Wireless
  Communications}, vol.~25, pp. 6869--6886, 2026.

\bibitem{matpencil}
W.~Li, H.~Yin, Z.~Qin, Y.~Cao, and M.~Debbah, ``A multi-dimensional matrix
  pencil-based channel prediction method for massive {MIMO} with mobility,''
  \emph{IEEE Transactions on Wireless Communications}, vol.~22, no.~4, pp.
  2215--2230, 2023.

\bibitem{HF_skywave}
D.~Shi, L.~Song, W.~Zhou, X.~Gao, C.-X. Wang, and G.~Ye~Li, ``Channel
  acquisition for {HF} skywave massive {MIMO}-{OFDM} communications,''
  \emph{IEEE Transactions on Wireless Communications}, vol.~22, no.~6, pp.
  4074--4089, 2023.

\bibitem{JCEP}
Y.~Zhu, J.~Zhuang, G.~Sun, H.~Hou, L.~You, and W.~Wang, ``Joint channel
  estimation and prediction for massive {MIMO} with frequency hopping
  sounding,'' \emph{IEEE Transactions on Communications}, vol.~73, no.~7, pp.
  5139--5154, 2025.

\bibitem{dda-net}
\BIBentryALTinterwordspacing
Y.~Ma, X.~Zhu, and T.~Li, ``{DDA}-{N}et: Accurate {TDD} channel estimation via
  deep unfolding the {D}oppler-delay-angle representation of channel signals,''
  2026. [Online]. Available: \url{https://arxiv.org/abs/2604.05389}
\BIBentrySTDinterwordspacing

\bibitem{ASM}
X.~Zhu, Y.~Ma, X.~Li, and T.~Li, ``Alternating subspace method for sparse
  recovery of signals,'' \emph{IEEE Transactions on Signal Processing},
  vol.~74, pp. 2175--2191, 2026.

\bibitem{LSCOA}
E.~K. Ryu and W.~Yin, \emph{Large-Scale Convex Optimization: Algorithms \&
  Analyses via Monotone Operators}.\hskip 1em plus 0.5em minus 0.4em\relax
  Cambridge University Press, Dec. 2022.

\bibitem{ADMM:Boyd}
S.~Boyd, N.~Parikh, E.~Chu, B.~Peleato, and J.~Eckstei, ``Distributed
  optimization and statistical learning via the alternating direction method of
  multipliers,'' \emph{Foundations and Trends in Machine Learning}, vol.~3, pp.
  1--122, 2010.

\bibitem{NIPS2010_2d6cc4b2}
P.~Garrigues and B.~Olshausen, ``Group sparse coding with a {L}aplacian scale
  mixture prior,'' in \emph{Advances in Neural Information Processing Systems},
  J.~Lafferty, C.~Williams, J.~Shawe-Taylor, R.~Zemel, and A.~Culotta, Eds.,
  vol.~23.\hskip 1em plus 0.5em minus 0.4em\relax Curran Associates, Inc.,
  2010.

\bibitem{TIPcrL}
S.~Zhang, Y.~Liu, X.~Li, and D.~Hu, ``Enhancing {ISAR} image efficiently via
  convolutional reweighted $\ell$1 minimization,'' \emph{IEEE Transactions on
  Image Processing}, vol.~30, pp. 4291--4304, 2021.

\bibitem{JunFang2015}
J.~Fang, Y.~Shen, H.~Li, and P.~Wang, ``Pattern-coupled sparse {B}ayesian
  learning for recovery of block-sparse signals,'' \emph{IEEE Transactions on
  Signal Processing}, vol.~63, no.~2, pp. 360--372, 2015.

\bibitem{803501}
M.~Speth, S.~Fechtel, G.~Fock, and H.~Meyr, ``Optimum receiver design for
  wireless broad-band systems using {OFDM}. i,'' \emph{IEEE Transactions on
  Communications}, vol.~47, no.~11, pp. 1668--1677, 1999.

\bibitem{salim2014channel}
O.~H. Salim, A.~A. Nasir, H.~Mehrpouyan, W.~Xiang, S.~Durrani, and R.~A.
  Kennedy, ``Channel, phase noise, and frequency offset in {OFDM} systems:
  Joint estimation, data detection, and hybrid {C}ramer-{R}ao lower bound,''
  \emph{IEEE Transactions on Communications}, vol.~62, no.~9, pp. 3311--3325,
  2014.

\bibitem{booktse}
D.~Tse and P.~Viswanath, \emph{Fundamentals of wireless communication}.\hskip
  1em plus 0.5em minus 0.4em\relax Cambridge university press, 2005.

\bibitem{bookUPA}
H.~Asplund, J.~Karlsson, F.~Kronestedt, E.~Larsson, D.~Astely, P.~von
  Butovitsch, T.~Chapman, M.~Frenne, F.~Ghasemzadeh, M.~Hagstr{\"o}m
  \emph{et~al.}, \emph{Advanced antenna systems for 5G network deployments:
  bridging the gap between theory and practice}.\hskip 1em plus 0.5em minus
  0.4em\relax Academic Press, 2020.

\bibitem{3GPP38901}
{3rd Generation Partnership Project (3GPP)}, ``Study on channel model for
  frequencies from 0.5 to 100 {GHz},'' Technical Report TR 38.901, Dec. 2019,
  version 16.1.0, Release 16.

\bibitem{ISTA1}
I.~Daubechies, M.~Defrise, and C.~De~Mol, ``An iterative thresholding algorithm
  for linear inverse problems with a sparsity constraint,''
  \emph{Communications on Pure and Applied Mathematics}, vol.~57, no.~11, pp.
  1413--1457, 2004.

\bibitem{Benner2009ADI}
P.~Benner, R.-C. Li, and N.~Truhar, ``On the {ADI} method for {S}ylvester
  equations,'' \emph{Journal of Computational and Applied Mathematics}, vol.
  233, no.~4, pp. 1035--1045, 2009.

\bibitem{ADI}
A.~Lu and E.~Wachspress, ``Solution of {L}yapunov equations by alternating
  direction implicit iteration,'' \emph{Computers \& Mathematics with
  Applications}, vol.~21, no.~9, pp. 43--58, 1991.

\bibitem{Simoncini2016MatrixEquations}
V.~Simoncini, ``Computational methods for linear matrix equations,'' \emph{SIAM
  Review}, vol.~58, no.~3, pp. 377--441, 2016.

\bibitem{pmlr-v32-zhong14}
W.~Zhong and J.~Kwok, ``Fast stochastic alternating direction method of
  multipliers,'' in \emph{Proceedings of the 31st International Conference on
  Machine Learning}, ser. Proceedings of Machine Learning Research, E.~P. Xing
  and T.~Jebara, Eds., vol.~32, no.~1.\hskip 1em plus 0.5em minus 0.4em\relax
  Bejing, China: PMLR, 22--24 Jun 2014, pp. 46--54.

\bibitem{siam:StochasticPDHG}
A.~Chambolle, M.~J. Ehrhardt, P.~Richt\'{a}rik, and C.-B. Sch\"{o}nlieb,
  ``Stochastic primal-dual hybrid gradient algorithm with arbitrary sampling
  and imaging applications,'' \emph{SIAM Journal on Optimization}, vol.~28,
  no.~4, pp. 2783--2808, 2018.

\bibitem{GSAng}
A.~Liu, V.~K.~N. Lau, and W.~Dai, ``Exploiting burst-sparsity in massive {MIMO}
  with partial channel support information,'' \emph{IEEE Transactions on
  Wireless Communications}, vol.~15, no.~11, pp. 7820--7830, 2016.

\bibitem{STCS}
L.~Chen, A.~Liu, and X.~Yuan, ``Structured turbo compressed sensing for massive
  {MIMO} channel estimation using a {M}arkov prior,'' \emph{IEEE Transactions
  on Vehicular Technology}, vol.~67, no.~5, pp. 4635--4639, 2018.

\bibitem{SchmidtMUSIC}
R.~Schmidt, ``Multiple emitter location and signal parameter estimation,''
  \emph{IEEE Transactions on Antennas and Propagation}, vol.~34, no.~3, pp.
  276--280, 1986.

\bibitem{Quadriga}
S.~Jaeckel, L.~Raschkowski, K.~Börner, and L.~Thiele, ``Quadriga: A 3-{D}
  multi-cell channel model with time evolution for enabling virtual field
  trials,'' \emph{IEEE Transactions on Antennas and Propagation}, vol.~62,
  no.~6, pp. 3242--3256, 2014.

\bibitem{QNOMP}
Y.~Zeng, M.~Han, X.~Li, and T.~Li, ``{QNOMP}: A quasi-newton-based joint
  multipath optimization method for super-resolution channel estimation and
  extrapolation,'' \emph{IEEE Access}, vol.~13, pp. 212\,652--212\,671, 2025.

\bibitem{OMP}
J.~A. Tropp and A.~C. Gilbert, ``Signal recovery from random measurements via
  orthogonal matching pursuit,'' \emph{IEEE Transactions on Information
  Theory}, vol.~53, no.~12, pp. 4655--4666, Dec. 2007.

\bibitem{beamforming}
K.~Mukkavilli, A.~Sabharwal, E.~Erkip, and B.~Aazhang, ``On beamforming with
  finite rate feedback in multiple-antenna systems,'' \emph{IEEE Transactions
  on Information Theory}, vol.~49, no.~10, pp. 2562--2579, 2003.

\bibitem{C-P}
A.~Chambolle and T.~Pock, ``A first-order primal-dual algorithm for convex
  problems with applications to imaging,'' \emph{Journal of mathematical
  imaging and vision}, vol.~40, no.~1, pp. 120--145, 2011.

\end{thebibliography}

\section*{Biography Section}
% \par\vspace{-20ex}
\begin{IEEEbiographynophoto}{Xu Zhu} received the BS degree in mathematics in 2021 from Peking University, where he is currently pursuing the Ph.D. degree in mathematics with an emphasis on optimization, signal processing and wireless communication.
\end{IEEEbiographynophoto}
% \par\vspace{-20ex}
\begin{IEEEbiographynophoto}{Tiejun Li} received the BS and MS degrees in mathematics from Tsinghua University, Beijing, in 1995 and 1998, respectively, and the PhD degree in mathematics from Peking University, in 2001. He is currently a Boya Distinguished Professor in the School of Mathematical Sciences at Peking University.  His research interest is the stochastic modeling and algorithms in science and engineering, which includes the stochastic models in signal processing, model reduction of complex networks, chemical reaction kinetics, rare events, and biological data analysis.
\end{IEEEbiographynophoto}
% \par\vspace{-3ex}

\end{document}